\documentclass[11pt]{article}
\usepackage[margin=1in]{geometry}
\usepackage{epsfig}
\usepackage{amssymb,amsmath,amsthm}
\usepackage{mathrsfs}
\usepackage{graphicx}
\usepackage{multirow}
\usepackage{color}
\usepackage{url}

\newcommand{\remove}[1]{}

\newcommand{\cD}{\mathcal{D}}
\newcommand{\cG}{\mathcal{G}}
\newcommand{\cE}{\mathcal{E}}

\newcommand{\ext}{\mathsf{Ext}}
\newcommand{\zo}{\{0,1\}}
\newcommand{\eps}{\varepsilon}

\newtheorem{thm}{Theorem}
\newtheorem{prop}[thm]{Proposition}
\newtheorem{claim}[thm]{Claim}

\theoremstyle{definition}

\newtheorem{theorem}[thm]{Theorem}
\newtheorem{lemma}[thm]{Lemma}
\newtheorem{definition}[thm]{Definition}

\newtheorem{corollary}[thm]{Corollary}

\newtheorem{remark}[thm]{Remark}

\newcommand*\samethanks[1][\value{footnote}]{\footnotemark[#1]}

\title{Secret Sharing with Binary Shares} 

\author{Fuchun Lin\thanks{Division of Mathematical Sciences, School of Physical and Mathematical Sciences, Nanyang Technological University, SG}
\and Mahdi Cheraghchi\thanks{Department of Computing, Imperial College London, UK} 
\and Venkatesan Guruswami\thanks{Computer Science Department, Carnegie Mellon University, USA}  
\and Reihaneh Safavi-Naini\thanks{Department of Computer Science, University of Calgary, CA}
\and Huaxiong Wang\samethanks[1]}

\usepackage[hidelinks]{hyperref}

\begin{document}

\maketitle

\begin{abstract}
Shamir's celebrated secret sharing scheme provides an efficient method for encoding a secret of arbitrary length $\ell$ among any $N \leq 2^\ell$ players such that for a threshold parameter $t$, (i) the knowledge of any $t$ shares does not reveal any information about the secret and, (ii) any choice of $t+1$ shares fully reveals the secret. It is known that any such threshold secret sharing scheme necessarily requires shares of length $\ell$, and in this sense Shamir's scheme is optimal. The more general notion of ramp schemes requires the reconstruction of secret from any $t+g$ shares, for a positive integer gap parameter $g$. Ramp secret sharing scheme necessarily requires shares of length $\ell/g$. Other than the bound related to secret length $\ell$, the share lengths of ramp schemes can not go below a quantity that depends only on the gap ratio $g/N$.

In this work, we study secret sharing in the extremal case of bit-long shares and arbitrarily small gap ratio $g/N$, where standard ramp secret sharing becomes impossible. We show, however, that a slightly relaxed but equally effective notion of semantic security for the secret, and negligible reconstruction error probability, eliminate the impossibility. Moreover, we provide explicit constructions of such schemes. One of the consequences of our relaxation is that, unlike standard ramp schemes with perfect secrecy, adaptive and non-adaptive adversaries need  different analysis and construction. For non-adaptive adversaries, we explicitly construct secret sharing schemes that  provide secrecy against any $\tau$ fraction of observed shares, and reconstruction from any $\rho$ fraction of shares, for any choices of $0 \leq \tau < \rho \leq 1$. Our construction achieves secret length $N(\rho-\tau-o(1))$, which we show to be optimal. For adaptive adversaries, we construct explicit schemes attaining a secret length $\Omega(N(\rho-\tau))$. We discuss our results and open questions.
\end{abstract}

\section{Introduction}

\newcommand{\mnote}[1]{{{\color{red}[Mahdi]#1}}}
\newcommand{\FF}{\mathbb{F}}
Secret sharing, introduced independently by Blakley \cite{Blakley} and Shamir \cite{Shamir},
is one of the most fundamental cryptographic primitives with far-reaching
applications, such as being a major tool in secure multiparty computation (cf.\ \cite{ref:CDB15}). 
The general goal in secret sharing
is to encode a secret $\mathsf{s}$ into a number of \emph{shares} 
$\mathsf{X}_1, \ldots, \mathsf{X}_N$ that are distributed
among $N$ players such that only certain \emph{authorized subsets}
of the players can reconstruct the secret. 
An \emph{authorized} subset of players is a set $A \subseteq [N]$ such that
the set of shares with indices in $A$ can collectively be used to reconstruct the secret $\mathsf{s}$  ({\em perfect reconstructiblity}).
On the other hand, $A$ is an \emph{unauthorized} subset if the knowledge of
the shares with indices in $A$ reveals no information about the secret ({\em perfect privacy}).
The set  of 
 authorized and unauthorized sets define an  access structure,
of which the most widely used is the so-called
\emph{threshold structure}. A  
secret sharing scheme with {\em threshold} access structure, is defined
with respect to an integer parameter $t$ and satisfies the following properties.
Any set $A \subseteq [N]$ with $|A| \leq t$ is an unauthorized set.
That is, the knowledge of any $t$ shares, or fewer, does not reveal any information
about the secret. On the other hand, any set $A$ with $|A| > t$ is an authorized set.
That is, the knowledge of any $t+1$ or more shares completely reveals the secret.

Shamir's secret sharing scheme \cite{Shamir} gives an elegant construction for the threshold access structure that can be interpreted 
as   the  use of Reed-Solomon codes for encoding the secret.
Suppose the  secret
$\mathsf{s}$ is an $\ell$-bit string and $N \leq 2^\ell$. Then, Shamir's scheme
treats the secret as an element of the finite field $\FF_q$, where $q = 2^\ell$,
padded with $t$ uniformly random and independent elements from the same field.
The resulting vector over $\FF_q^{t+1}$ is then encoded using 
a Reed-Solomon
code of length $N$, providing $N$ shares of length $\ell$ bits each.
The fact that a Reed-Solomon code is Maximum Distance Separable (MDS) can then be used to show that 
the threshold guarantee for privacy and reconstruction is satisfied.

Remarkably, Shamir's scheme is optimal for threshold secret sharing in the
following sense: Any threshold secret sharing scheme  sharing $\ell$-bit 
secrets necessarily requires shares of length at least $\ell$, and Shamir's
scheme attains this lower bound \cite{ref:Sti92}.

It is natural to ask whether secret sharing is possible at
share lengths below the secret length $\log q<\ell$, where $\log$ is to base $2$ throughout this work. 
Of course, in this case, the threshold guarantee  that requires all subsets of participants be either authorized, or unauthorized, can no longer
be attained. 
Instead, the notion can be relaxed to \emph{ramp}
secret sharing which  allows some subset of participants to learn some information about the secret. 
A ramp scheme is defined with respect to two threshold
parameters, $t$ and $r > t+1$. As in threshold scheme, the knowledge 
of any $t$ shares or fewer does not reveal any information about the secret.
On the other hand, any $r$ shares can be used to reconstruct the secret. The subsets of size between $t+1$ and $r-1$, may learn some information  about the secret. The information-theoretic bound (see e.g. \cite{OKT92}) now becomes 
\begin{equation}\label{eq: ITbound}
\ell\leq (r-t)\log q.
\end{equation}
Ideally, one would like to obtain equality in (\ref{eq: ITbound}) for as general parameter settings as possible. 

Let $g\colon=r-t$ denote the {\em gap} between the privacy and reconstructibility parameters. Let the secret length $\ell$ and the number of players $N$ be unconstrained integer parameters.
It is known that, using Reed-Solomon code interpretation of Shamir's approach applied to
a random linear code, for every fixed {\em relative gap} $\gamma\colon=g/N$,
there is a constant $q$ only depending on $\gamma$ such that
a ramp secret sharing scheme with share size $q$ exists.
Such schemes can actually be constructed by using explicit
algebraic geometry codes instead of random linear codes. 
In fact, this dependence of share size $q$ on relative gap $g/N$ is inherent for threshold and more generally ramp schemes. 
It is shown in an unpublished work of Kilian and Nisan
\footnote{Their result is unpublished, and is independently obtained and generalised by \cite{ref:CCX13} (see \cite[Appendix A]{ref:CCX13}).} 
for threshold schemes, and later more generally in \cite{ref:CCX13}, that for ramp schemes with share size $q$, threshold gap $g$, privacy threshold $t$ and unconstrained number of players $N$, the following always holds:
$
q\geq(N-t+1)/g.
$
Very recently in \cite{BGK16}, a new bound with respect to the reconstruction parameter $r$ is proved through connecting secret sharing for one bit secret to game theory:
$
q\geq(r+1)/g.
$
These two bounds together yield
\begin{equation}\label{eq: bound}
q\geq(N+g+2)/(2g). 
\end{equation}
Note that the bound (\ref{eq: bound}) is very different from the bound (\ref{eq: ITbound}) in nature. The bound (\ref{eq: ITbound}) is the fundamental limitation of information-theoretic security, bearing the same flavour as the One-Time-Pad. The bound (\ref{eq: bound}) is independent of the secret length and holds even when the secret is one bit. 

We ask the following question: {\em For a fixed constant share size $q$ (in particular, $q=2$), is it possible to construct (relaxed but equally effective) ramp secret sharing schemes with arbitrarily small relative gap $\gamma>0$ that asymptotically achieve equality in} (\ref{eq: ITbound})?

Our results in this work show that the restriction (\ref{eq: bound}) can be overcome if
we allow a negligible privacy error in statistical distance (semantic security) and
a negligible reconstruction error probability. 

\subsection{Related work}


Secret sharing with binary shares is recently selectively 
studied in \cite{BIVW16} with the focus on the tradeoff between the privacy parameter $t$ and the computational complexity of the reconstruction function. The concern of this line of works is to construct secret sharing schemes whose sharing and reconstructing algorithms are both in the complexity class AC$^0$ (i.e., constant depth circuits). 

The model of secret sharing considered in \cite{BIVW16} has perfect privacy, namely, the distributions of any $t$ shares from a pair of distinct secrets are identical, while reconstruction is not necessarily with probability $1$. In a followup work \cite{BW17}, the privacy is further relaxed to semantic security with an error parameter $\varepsilon>0$. The relaxation is shown to be useful in allowing more choices of $t$ while keeping the computational complexity of the reconstruction algorithm within AC$^0$.

In \cite{BIVW16,BW17}, the secret is one bit. In \cite{AC0}, secrets of length equal to a fraction of $N$ (number of players) is considered. This time binary secret sharing with adaptive and non-adaptive adversaries similar to the model we consider in this work is defined. However the paper considers only a privacy threshold $t$, and reconstruction is from the full share set ($r=N$ always). Their goal is to achieve large secrets $\ell=\Omega(N)$ over binary shares with large privacy parameter $t=\Omega(N)$, which is also similar to ours. They have an additional goal of keeping the computational complexity of the reconstruction algorithm within AC$^0$, which we do not consider in this work. 
Their large privacy parameter $t=\tau N$ is with a $\tau$ much smaller than $1$, which means that the relative threshold gap $\gamma=1-\tau$ can not be arbitrarily small. 


In the literature, perhaps the closest notion to secret sharing with binary shares is
that of \emph{wiretap codes}, first  studied in information theory. 
In the basic 
 wiretap channel model of Wyner and its extension to {\em broadcast channel with confidential messages} \cite{Wyner,CK78},
there is a point-to-point \emph{main} channel between a sender
and a receiver that has partial leakage to the adversary, and the leakage of information 
is modelled by a second 
point-to-point \emph{wiretapper} 
channel between the sender and the adversary. To goal of the sender  is
to encode messages 
 in such a way that the receiver can decode them, while 
 the adversary does not learn much about them
\cite{Wyner}. 
The highest information rate achievable for a wiretap channel is called the {\em secrecy capacity}. 
A Binary Erasure Channel with erasure probability $p$ ($\mathsf{BEC}_p$) is a probabilistic transformation that maps inputs $0$ and $1$ to a symbol $?$, which denotes erasure, with probability $p$ and to the inputs themselves with probability $1-p$. The erasure channel scenario of wiretap model is closely related to secret sharing with fixed share size.
For a pair ($\mathsf{BEC}_{p_m}$,$\mathsf{BEC}_{p_w}$) of BEC's, such that $p_m<p_w$, it is known that the secrecy capacity is the difference of the respective channel capacities: $(1-p_m)-(1-p_w)=p_w-p_m$. 

It is important to note the distinctions between the
erasure wiretap model above, 
and binary secret sharing.
First, 
the guarantees
of wiretap codes are required to only hold for random 
messages, whereas in secret sharing, the cryptographic
convention of security for worst-case messages is required.
Second, in the standard wiretap model, the notion of secrecy
is information-theoretic and is typically measured
in terms of the mutual information. 
Namely, for a random $\ell$-bit message
$\mathsf{M}$, and letting $\mathsf{W}$ denote the information delivered to
the adversary,
secrecy is satisfied in the weak (resp., strong) sense if
$\mathsf{I}(\mathsf{M}; \mathsf{W}) \leq \eps \ell$ (resp., $\mathsf{I}(\mathsf{M};\mathsf{W}) \leq \eps$) 
for an arbitrarily small constant $\eps$.
The randomness of the random variable $\mathsf{W}$ depends on the
randomness of $\mathsf{M}$, the two channels, as well as the internal
randomness of the encoder.
For secret sharing, on the other hand, either perfect secrecy
or semantic secrecy (negligible leakage with respect to statistical distance)
is a requirement.

The notion of secrecy in wiretap codes has evolved over years. More recently the notion of semantic security for wiretap model has been introduced \cite{btv12}, which allows arbitrary  message distribution 
and is shown to be equivalent to negligible leakage with respect to statistical distance. 

There remains one last distinction between {\em semantically secure} wiretap model and secret sharing with fixed share size. 
That is the nature of the main and wiretapper channels
are typically stochastic (e.g., the erasure channel
with random i.i.d. erasures), whereas for secret sharing
a worst-case guarantee for the erasure patterns
is required. Namely, in secret sharing, reconstruction with overwhelming probability is required for every
choice of $r$ or more shares, 
and privacy of the secret is required
for every (adaptive or non-adaptive) choice of the 
$t$ shares observed by the adversary. 

On the other hand, Ozarow and Wyner proposed 
the {\em wiretap channel II}, where an adversary observes arbitrary $t$ out of the total $N$ bits of the communication \cite{OW84}. The wiretapper channel of the wiretap channel II is the adversarial analogue of a $\mathsf{BEC}_{p_w}$ with erasure probability ${p_w}=(N-t)/N$. This is exactly the same as the privacy requirement of the binary secret sharing. But in the wiretap channel II model, the main channel is a clear channel. This is corresponding to the special case of binary secret sharing where the reconstruction is only required when all shares are available. The adversarial analogue of the erasure scenario Wyner wiretap channel should have a main channel that erases $Np_m$ components out of the total $N$ components, which is the same as the reconstruction requirement of binary secret sharing.
Moreover, same as the Wyner wiretap channel model, the secrecy of the wiretap channel II is only required for uniform message and satisfies the weak (resp., strong) secrecy mentioned above.

\subsection{Our contributions} 

We motivate the study of secret sharing scheme with 
fixed share size $q$, and study the extremal case of binary shares.
Our goal is to show that even in this extremely restrictive case, 
a slight relaxation of the privacy and reconstruction notions of ramp 
 secret sharing 
 guarantees  explicit construction of families of ramp schemes 
 \footnote{We abuse the notion and will continue calling these relaxed schemes ramp schemes.}
 with any constant relative privacy and reconstruction thresholds $0 \leq \tau < \rho \leq 1$, in particular, the relative threshold gap $\gamma=\rho-\tau$ can be an arbitrarily small constant.
Namely, for any constants
$0 \leq \tau < \rho \leq 1$, it can 
be guaranteed that
any $\tau N$ or fewer shares reveal essentially no information
about the secret, whereas any $\rho N$ or more shares can
reconstruct the exact secret with a negligible failure probability. 
While we only focus on the extremal special case $q=2$ in this presentation, all our results can be extended to any constant $q$
(see Section~\ref{sec: conclusion}).

We consider binary sharing of a large $\ell$-bit secret 
and for this work focus on the asymptotic case where 
the secret length $\ell$, and consequently the number of players $N$,
are sufficiently large. 
We replace perfect privacy with 
 semantic security,  the strongest cryptographic notion of privacy second only to perfect privacy.
That is, for any two 
secrets (possibly chosen by the adversary), we require the
 adversary's view 
 to be statistically indistinguishable. The view of the adversary is a random variable with randomness coming solely from the internal
randomness of the sharing algorithm. 
The notion of indistinguishability that we use
is statistical (total variation) distance bounded
by a leakage error parameter $\eps$ that is negligible in $N$.
Using  non-perfect privacy creates a
distinction between non-adaptive and adaptive
secrecy. A non-adaptive adversary 
chooses any $\tau$ fraction of the $N$ players at once, 
and receives their corresponding shares. 
An adaptive adversary however, selects share holders one by one, receives their shares and uses its   available information to make its next choice.
When $\eps=0$, i.e., when perfect privacy holds, non-adaptive
secrecy automatically implies adaptive secrecy as well.
However, this is not necessarily the case when $\eps>0$ 
and we thus study the two cases separately. Similarly, we replace the perfect reconstruction with probabilistic reconstruction allowing a failure probability $\delta$ that is negligible in $N$. The special case of $\delta=0$ means perfect reconstruction.

Note that secret sharing with fixed share size
necessarily imposes certain restrictions that are not
common in standard secret sharing. 
Unlike secret sharing with share length dependent on the secret length (for threshold schemes) or secret length and threshold gap (for ramp schemes), 
binary sharing of an $\ell$-bit secret obviously requires
at least $\ell$ shares to accommodate 
the secret information. 
 For a family of ramp secret sharing schemes with fixed share size $q$ and fixed relative thresholds $0 \leq \tau < \rho \leq 1$, as $N$ grows the 
 absolute gap length $(\rho-\tau)N$ grows, and the accommadatable length of the secret is expected to grow and so
 the 
 ratio $\ell/N \in (0,1]$ 
becomes a key parameter of interest for the family, referred to as the {\em coding rate}. As is customary 
in coding theory, it is desired to characterize
the maximum possible ratio $\ell/N \in (0,1]$ for binary secret sharing. 
We use the relation (a similar relation was used in \cite{nearoptimalwt} for robust secret sharing) between a binary secret sharing family with 
relative threshold $(\tau,\rho)$ and codes for a Wyner wiretap channel
with two BEC's to derive a coding rate upper
bound of $\rho-\tau$ for binary secret sharing (see Lemma~\ref{lem: upperbound}).

Our main technical contributions are explicit constructions
of binary secret sharing schemes in both the non-adaptive
and adaptive models, and proving optimality of non-adaptive construction. Namely, we prove the following:

\begin{theorem}[informal summary of Lemma~\ref{lem: upperbound},
Corollary~\ref{cor: non-adaptive}, and Corollary~\ref{cor: adaptive}]
For any choice of $0 \leq \tau < \rho \leq 1$,
and large enough $N$, there is an explicit construction of a
binary secret sharing scheme 
with $N$ players that provides (adaptive or non-adaptive)
privacy against leakage of any $\tau N$ or fewer shares,
as well as reconstruction from any $\rho N$ or more of the shares
(achieving semantic secrecy with negligible error and imperfect reconstruction with
negligible failure probability).
For non-adaptive secrecy, the scheme shares a secret of length
$\ell = (\rho-\tau-o(1))N$, which is asymptotically
optimal. For adaptive secrecy, the scheme shares a secret of length
$\ell =\Omega((\rho-\tau)N)$. 
\end{theorem}

As a side contribution, our findings unify the Wyner wiretap model and its adversarial analogue. 
Our capacity-achieving construction of binary secret sharing for non-adaptive adversaries implies that the secrecy capacity of the adversarial analogue of the erasure scenario Wyner wiretap channel is similarly characterized by the erasure ratios of the two channels. Moreover, the secrecy can be strengthened to semantic security.

This answers an open question posted in \cite{ALCP09}. The authors studied a generalisation of the wiretap II model, where the adversary chooses $t$ bits to observe and erases them. They showed that the rate $1-\tau-h_2(\tau)$, where $h_2(\cdot)$ is the binary entropy function, 
can be achieved and left open the question of whether a higher rate is achievable. Our result specialized to their setting shows that, the rate $1-2\tau$ can be explicitly achieved.

\subsection{Our approach and techniques}

Our explicit constructions follow the paradigm of invertible
randomness extractors formalized in~\cite{CDS12}.
Invertible extractors were used in \cite{CDS12}
for explicit construction of optimal wiretap coding schemes
in the Wiretap channel II \cite{OW84}.
This, in particular, is corresponding to the $\rho=1$ special case of 
secret sharing 
where reconstruction is only required when
all shares are available. Moreover, the secrecy there is an information-theoretic notion, and only required to hold for uniform messages. The consequence of the latter is that the construction in \cite{CDS12} does not directly give us binary secret sharing, not even for the $\rho=1$ special case. The exposition below is first focused on how semantic security is achieved.  

As in \cite{CDS12}, we rely on
{\em invertible affine extractors} as our primary 
technical tool. 
Such an extractor is
an explicit function $\mathsf{AExt}\colon \zo^n \to \zo^\ell$ such
that, for any random variable $\mathsf{X}$ uniformly distributed
over an unknown $k$-dimensional affine subspace of $\FF_2^n$,
the distribution of $\mathsf{AExt}(\mathsf{X})$ is close to the uniform distribution over $\FF_2^\ell$
in statistical distance. Furthermore, the invertibility guarantee
provides an efficient algorithm for 
sampling a uniform element from the set $\mathsf{AExt}^{-1}(\mathsf{s})$ of pre-images 
for any given output $\mathsf{s}\in\FF_2^\ell$. 

It is then natural to consider the affine extractor's uniform inverter 
as a candidate building block for the sharing algorithm of a secret sharing scheme. 
Intuitively, if the secret $\mathsf{s}$ is chosen
uniformly at random, we have the guarantee that
for any choice of a bounded number of
the bits of its random pre-image revealed to the adversary, the distribution of the random pre-image conditioned on the revealed value satisfies that of an affine source. 
Now according to the definition of an affine extractor,
the extractor's
output (i.e., the secret $\mathsf{s}$) remains uniform
(and thus unaffected in distribution)
given the information revealed to the adversary. 
Consequently, secrecy should at least hold in an information-theoretic
sense, i.e. the mutual information between the secret and the revealed vector components is zero. This is what was formalized and used in~\cite{CDS12} for the construction
of Wiretap channel II codes.

For non-adaptive adversaries, in fact it is possible
to use \emph{invertible seeded extractors} rather than 
invertible affine extractors described in the above construction. 
A (strong) seeded extractor
assumes, in addition to the main input,
an independent seed as an auxiliary input and 
ensures uniformity of the output for most fixings
of the seed. The secret sharing encoder appends a 
randomly chosen seed to the encoding and inverts the extractor
with respect to the chosen seed. Then, the above
argument would still hold even if the seed is completely
revealed to the adversary.

The interest in the use
of seeded, as opposed to seedless affine, extractors
is twofold. First, nearly optimal and very efficient
constructions of seeded extractors are known in the
literature that extract nearly the entire source entropy
with only a short seed. This allows us to attain
nearly optimal rates for the non-adaptive case. 
Furthermore, and crucially, such nearly optimal extractor constructions
(in particular, Trevisan's extractor \cite{Trevisan,improvement})
can in fact be linear functions for every fixed choice
of the seed (in contrast, seedless affine extractors
can never be linear functions). We take advantage of the linearity of the
extractor in a crucial way and use a rather
delicate analysis to show that in fact the linearity
of the extractor can be utilized to prove that the
resulting secret sharing scheme provides the stringent
worst-case secret guarantee which is a key requirement
distinguishing secret sharing schemes (a cryptographic primitive)
from wiretap codes (an information-theoretic notion).


Using a seeded extractor instead of a seedless extractor, however, introduces a new challenge. In order for the seeded extractor to work, the seed has to be independent of the main input, which is a distribution induced by the adversary's choice of reading positions. The independence of the seed and the main input can be directly argued when the adversary is non-adaptive. An adaptive adversary, however, may choose its reading positions to learn about the seed first, and then choose the rest of the reading positions according the value of the seed. In this case, we can not prove the independence of the seed and the main input. 

For adaptive adversaries, we go back to using an invertible affine extractor.
We prove that both security
for worst-case messages and against adaptive 
adversaries are guaranteed if the affine extractor 
provides the strong guarantee of having a
nearly uniform output with respect to the $\ell_\infty$ 
measure rather than $\ell_1$. 
However, this comes at the  cost of the extractor not being able
to extract the entire  entropy of the source, leading to ramp
secret sharing schemes with slightly sub-optimal rates,
albeit still achieving rates within a constant
factor of the optimum. 
As a proof of concept, we utilize
a simple padding and truncation technique
to convert any off-the-shelf seedless affine extractor
(such as those of Bourgain \cite{Bourgain} or Li \cite{XinLiAffine})
to one that satisfies the stronger uniformity condition
that we require. 

We now turn to reconstruction from an incomplete set of shares.
In order to provide reconstructibility from a subset
of size $r$ of the shares, we naturally compose the
encoding obtained from the extractor's inversion routine
with a linear erasure-correcting code. The linearity of
the code ensures that the extractor's input
subject to the adversary's observation (which now can
consist of linear combinations of the original encoding) remains
uniform on some affine space, thus preserving the
privacy guarantee. 

However, since by the known rate-distance trade-offs 
of binary error-correcting codes, no deterministic
coding scheme can correct more than a $1/2$ fraction of erasures
(a constraint that would limit the choice of $\rho$),
the relaxed notion of \emph{stochastic coding schemes}
is necessary for us to allow reconstruction for all choices of
$\rho \in (\tau,1]$.
 Intuitively, a stochastic code is a randomized encoder with a deterministic decoder, that allows the required fraction of errors to be corrected.
We utilize what we call a \emph{stochastic affine} code.
Such codes are equipped with encoders that are affine
functions of the message for every fixing of the encoder's
internal randomness. We show that such codes are
as suitable as deterministic linear codes for providing
the linearity properties that our construction needs.

In fact, we need capacity-achieving stochastic erasure 
codes, i.e., those that correct every $1-\rho$ fraction
of erasures at asymptotic rate $\rho$, to be able to construct binary secret sharing schemes with arbitrarily small relative gap $\gamma=\rho-\tau$.
To construct capacity-achieving stochastic affine erasure 
codes, we utilize
a construction of stochastic codes due to Guruswami
and Smith \cite{GS10} 
for  bit-flip errors.  We observe
that this construction can be modified 
to yield capacity-achieving
erasure codes. 
 Roughly speaking, this is
achieved by taking an explicit capacity-achieving linear code for BEC and pseudo-randomly
shuffling the codeword positions. Combined with a delicate
encoding of hidden ``control information'' to communicate
the choice of the permutation to the decoder in a robust
manner, the construction transforms robustness against random
erasures to worst-case erasures at the cost of making the
encoder randomized.


\subsection{Organization of the paper}
Section \ref{sec: pre} contains a brief introduction to the two building blocks for our constructions: randomness extractors and stochastic codes. In Section \ref{sec: basic SSS}, we formally define the binary secret sharing model and prove a coding rate upper bound. 
Section \ref{sec: capacity} contains a capacity-achieving construction with privacy against non-adaptive adversaries.
Section \ref{sec: construction} contains a constant rate construction with privacy against adaptive adversaries. Finally, we conclude the paper and discuss open problems in Section \ref{sec: conclusion}.

\section{Preliminaries and definitions} \label{sec: pre} 
In this section, we review the necessary facts and results about randomness extractors, both the seeded and seedless affine variants, as well as the stochastic erasure correcting codes.

Randomness extractors extract close to uniform bits from input sequences that 
are not uniform but have some guaranteed entropy. 
The closeness to uniform of the extractor output is measured by the statistical distance (half the $\ell_1$-norm). For a set $\mathcal{X}$, we use $\mathsf{X}\leftarrow\mathcal{X}$ to denote that $\mathsf{X}$ is distributed over the set $\mathcal{X}$. For two random variables $\mathsf{X},\mathsf{Y}\leftarrow\mathcal{X}$, the statistical distance between $\mathsf{X}$ and $\mathsf{Y}$ 
is defined as,
\begin{align*}
\mathsf{SD}(\mathsf{X};\mathsf{Y})= \dfrac{1}{2}\sum_{\mathsf{x} \in \mathcal{X}}\left |\mathsf{Pr}[\mathsf{X}=\mathsf{x}]-\mathsf{Pr}[\mathsf{Y}=\mathsf{x}]\right |.
\end{align*}
We say $\mathsf{X}$ and $\mathsf{Y}$ are $\varepsilon$-close 
if $\mathsf{SD}(\mathsf{X},\mathsf{Y})\leq \varepsilon$.
A {\em randomness source} is a random variable with  lower bound on its min-entropy, which is defined by $\mathsf{H}_\infty(\mathsf{X})=-\log \max_\mathsf{x}\{\mathsf{Pr}[\mathsf{X}=\mathsf{x}]\}$.
We say  a  random variable $\mathsf{X} \leftarrow \lbrace 0,1\rbrace^{n}$ is a {\em $(n,k)$-source}  if $\mathsf{H}_\infty(\mathsf{X})\geq k$. 

For well structured sources, there exist deterministic functions that can extract close to uniform bits. The {\em support} of $\mathsf{X}\leftarrow\mathcal{X}$ is the set of $\mathsf{x}\in\mathcal{X}$ such that $\mathsf{Pr}[\mathsf{X}=\mathsf{x}]>0$.
An  \textit{affine} $(n,k)$-source is an $(n,k)$-source whose support is an affine 
sub-space of $\lbrace0,1\rbrace^n$ and each vector in the support occurs with the same probability. Let $\mathsf{U}_m$ denote the random variable uniformly distributed over $\{0,1\}^m$.

\begin{definition}\label{def: AExt} A function $\mathsf{AExt}\colon\lbrace0,1\rbrace^n \to \lbrace 0,1 \rbrace ^m$ is an affine $(k, \varepsilon)$-extractor if for any affine $(n,k)$-source $\mathsf{X}$, we have 
\begin{align*}
\mathsf{SD}(\mathsf{AExt}(\mathsf{X});\mathsf{U}_m) \leq \varepsilon.
\end{align*} 
\end{definition}
An affine extractor can not be a linear function.


For general $(n,k)$-sources, there does not exist a deterministic function that can extract close to uniform bits from all of them simultaneously. A family of deterministic functions are needed.

\begin{definition}\label{def: strong extractor}
A function $\mathsf{Ext}\colon\lbrace0,1\rbrace^d\times \lbrace0,1\rbrace^n  \to \lbrace 0,1 \rbrace ^m$ is a strong seeded $(k, \varepsilon)$-extractor if for any $(n,k)$-source $\mathsf{X}$, we have 
\begin{align*}
\mathsf{SD}(\mathsf{S},\mathsf{Ext}(\mathsf{S},\mathsf{X});\mathsf{S},\mathsf{U}_m) \leq \varepsilon,
\end{align*} 
where $\mathsf{S}$ is chosen uniformly from $\lbrace0,1\rbrace^d$. 
A seeded extractor $\mathsf{Ext}(\cdot,\cdot)$ is called linear if for any fixed seed $\mathsf{S}=\mathsf{s}$, the function $\mathsf{Ext}(\mathsf{s},\cdot)$ is a linear function.
\end{definition}

We will use Trevisan's extractor \cite{Trevisan} in our first construction. In particular, we use the following improvement of this extractor due to Raz, Reingold and Vadhan \cite{improvement}.


\begin{lemma}[\cite{improvement}] \label{lem: Ext}
There is an explicit linear strong seeded $(k,\varepsilon)$-extractor $\mathsf{Ext}\colon\lbrace0,1\rbrace^d\times \lbrace0,1\rbrace^n  \to \lbrace 0,1 \rbrace ^{\ell}$ with $d=O(\log^3(n/\varepsilon))$ and $\ell=k-O(d)$.
\end{lemma}

We will use Bourgain's affine extractor in our second construction. We note,
however, that we could have used other explicit extractors for this purpose, such
as \cite{XinLiAffine}.
\begin{lemma}[\cite{Bourgain}]\label{lem: affine} 
For every constant $0<\mu\leq 1$, there is an explicit affine $(\mu n,\varepsilon)$-extractor $\mathsf{AExt}\colon \{0,1\}^n\rightarrow\{0,1\}^m$ with output length $m=\Omega(n)$ and error $\varepsilon=2^{-\Omega(n)}$.
\end{lemma}

Explicit constructions of randomness extractors have efficient forward direction of extraction. In some applications, 
we usually need to efficiently invert the process: Given an extractor output, sample a random pre-image. 

\begin{definition}[\cite{CDS12}]\label{def: invertible} 
Let $f$ be a mapping from $\{0,1\}^n$ to $\{0,1\}^m$. For $0\leq v< 1$, a function $\mathsf{Inv}\colon \{0,1\}^m\times\{0,1\}^r \rightarrow\{0,1\}^n$ is called a $v$-inverter for $f$ if the following conditions hold:
\begin{itemize}
\item (Inversion) Given $\mathsf{y}\in\{0,1\}^m$ such that its pre-image $f^{-1}(\mathsf{y})$ is nonempty, for every $\mathsf{r}\in\{0,1\}^r$ we have $f (\mathsf{Inv}(\mathsf{y}, \mathsf{r})) = \mathsf{y}$.
\item (Uniformity) $\mathsf{Inv}(\mathsf{U}_m,\mathsf{U}_r)$ is $v$-close to $\mathsf{U}_n$.
\end{itemize}
A $v$-inverter is called efficient if there is a randomized algorithm that runs in worst-case polynomial time and, given $\mathsf{y}\in\{0,1\}^m$ and $\mathsf{r}$ as a random seed, computes $\mathsf{Inv}(\mathsf{y}, \mathsf{r})$. We call a mapping $v$-invertible if it has an efficient $v$-inverter, and drop the prefix $v$ from the notation when it is zero. We abuse the notation and denote the inverter of $f$ by $f^{-1}$.
\end{definition}


\smallskip
\noindent
A {\em stochastic code} 
has a randomised encoder and a deterministic decoder. The encoder $\mathsf{Enc}\colon\{0,1\}^m\times\mathcal{R}\rightarrow \{0,1\}^n$ uses local randomness $\mathsf{R}\leftarrow\mathcal{R}$ to encode a message $\mathsf{m}\in\{0,1\}^m$. The decoder is a deterministic function $\mathsf{Dec}\colon\{0,1\}^n\rightarrow \{0,1\}^m\cup\{\bot\}$. The decoding probability is defined over the encoding randomness $\mathsf{R}\leftarrow\mathcal{R}$. 
Stochastic codes are known to explicitly achieve the capacity of some adversarial channels \cite{GS10}.

Affine sources play an important role in our constructions. We define a general requirement for the stochastic code used in our constructions. 
\begin{definition}[Stochastic Affine codes]\label{def: affine} Let $\mathsf{Enc}\colon\{0,1\}^m\times\mathcal{R}\rightarrow \{0,1\}^n$ be the encoder of a stochastic code. We say it is a stochastic affine code if for any $\mathsf{r}\in\mathcal{R}$, the encoding function $\mathsf{Enc}(\cdot,\mathsf{r})$ specified by $\mathsf{r}$ is an affine function  
of the message. That is we have 
\begin{align*}
\mathsf{Enc}(\mathsf{m},\mathsf{r})=\mathsf{m}G_\mathsf{r}+\Delta_\mathsf{r},
\end{align*}
where $G_\mathsf{r}\in\{0,1\}^{m\times n}$ 
and $\Delta_\mathsf{r}\in\{0,1\}^n$ are specified by 
the randomness $\mathsf{r}$.

\end{definition}

We then adapt a construction in \cite{GS10} to obtain the following capacity-achieving Stochastic Affine-Erasure Correcting Code (SA-ECC). 
In particular, we show for any $p\in[0,1)$, there is an explicit stochastic affine code that corrects $p$ fraction of adversarial erasures and achieves the rate $1-p$ (see Appendix~\ref{apdx: SA-ECC} for more details). 


\begin{lemma}[Adapted from \cite{GS10}]
\label{th: SA-ECC} 
For every $p\in[0,1)$, and every $\xi>0$, there is an efficiently encodable and decodable stochastic affine code $(\mathsf{Enc},\mathsf{Dec})$ with rate $R=1-p-\xi$ such that for every $\mathsf{m}\in\{0,1\}^{NR}$ and erasure pattern of at most $p$ fraction, we have $\mathsf{Pr}[\mathsf{Dec}(\widetilde{\mathsf{Enc}(\mathsf{m})})=\mathsf{m}]\geq 1-\exp(-\Omega(\xi^2N/\log^2 N))$, where $\widetilde{\mathsf{Enc}(\mathsf{m})}$ denotes the partially erased random codeword and $N$ denotes the length of the codeword.
\end{lemma}


\section{Binary secret sharing schemes} \label{sec: basic SSS} 

In this section, we define our model of nearly-threshold binary secret sharing schemes. We begin with a description of the two models of non-adaptive and adaptive adversaries which can access up to $t$ of the $N$ shares.

A {\em leakage oracle} is a machine $\mathcal{O}(\cdot)$ that takes as input an $N$-bit string $\mathsf{c}\in\{0,1\}^N$ and then answers the {\em leakage queries} of the type $I_j$, for $I_j\subset[N]$, $j=1,2,\ldots, q$. Each query $I_j$ is answered with $\mathsf{c}_{I_j}$. An interactive machine $\mathcal{A}$ that issues the leakage queries is called a {\em leakage adversary}. Let $A_\mathsf{c}=\cup_{j=1}^qI_j$ denote the union of all the index sets chosen by $\mathcal{A}$ when the oracle input is $\mathsf{c}$. 
The oracle is called $t$-bounded, denoted by $\mathcal{O}_t(\cdot)$, if it rejects leakage queries from $\mathcal{A}$ if there exists some $\mathsf{c}\in\{0,1\}^N$ such that $|A_\mathsf{c}|>t$. An adaptive leakage adversary decides the index set $I_{j+1}$ according to the oracle's answers to all previous queries $I_1,\ldots,I_j$. A non-adaptive leakage adversary has to decide the index set $A_\mathsf{c}$ before any information about $\mathsf{c}$ is given. This means that for a non-adaptive adversary, given any oracle input $\mathsf{c}\in\{0,1\}^N$, we always have $A_\mathsf{c}=A$ for some $A\subset[N]$. 
Let $\mathsf{View}^{\mathcal{O}_t(\cdot)}_\mathcal{A}$ denote the view of the leakage adversary $\mathcal{A}$ interacting with a $t$-bounded leakage oracle. 
When $\mathcal{A}$ is non-adaptive, we use the shorthand $\mathsf{View}^{\mathcal{O}_t(\cdot)}_\mathcal{A}=(\cdot)_A$, for some $A\subset[N]$ of size $|A|\leq t$. 

A function $\varepsilon\colon\mathbb{N}\rightarrow\mathbb{R}$ is called {\em negligible} if for every positive integer $k$, there exists an $N_k\in\mathbb{N}$ such that $|\varepsilon(N)|<\frac{1}{N^k}$ for all $N>N_k$. The following definition of ramp Secret Sharing Scheme (SSS) allows imperfect privacy and reconstruction with errors bounded by negligible functions $\varepsilon(\cdot)$ and $\delta(\cdot)$, respectively.

\begin{definition} \label{def: SSS} For any $0\leq\tau<\rho\leq1$, an $(\varepsilon(N),\delta(N))$-SSS with relative threshold pair $(\tau,\rho)$ is a pair of polynomial-time algorithms $(\mathsf{Share},\mathsf{Recst})$,
\begin{align*}
\mathsf{Share}\colon\{0,1\}^{\ell(N)}\times\mathcal{R}\rightarrow\{0,1\}^N,
\end{align*}
where $\mathcal{R}$ denote the randomness set, and
\begin{align*}
\mathsf{Recst}\colon\widetilde{\{0,1\}^N}\rightarrow\{0,1\}^{\ell(N)}\cup\{\bot\},
\end{align*}
where $\widetilde{\{0,1\}^N}$ denotes the subset of $(\{0,1\}\cup\{?\})^N$ with at least $N\rho$ components not equal to the erasure symbol ``$?$'', that satisfy the following properties.\begin{itemize}

\item Reconstruction: Given $r(N)=N\rho$ correct shares of a share vector $\mathsf{Share}(\mathsf{s})$, the reconstruct algorithm $\mathsf{Recst}$ reconstructs the secret $\mathsf{s}$ with probability at least $1-\delta(N)$. 

When $\delta(N)=0$, we say the SSS has perfect reconstruction.

\item Privacy (non-adaptive/adaptive): 

  \begin{itemize}
    \item Non-adaptive: for any $\mathsf{s}_0,\mathsf{s}_1\in\{0,1\}^{\ell(N)}$, any $A\subset [N]$ of size $|A|\leq t(N)=N\tau$,
      \begin{equation}\label{eq: WtII security}
      \mathsf{SD}(\mathsf{Share}(\mathsf{s}_0)_{A};\mathsf{Share}(\mathsf{s}_1)_{A})\leq \varepsilon(N).
      \end{equation}
    \item Adaptive: for any $\mathsf{s}_0,\mathsf{s}_1\in \{0,1\}^{\ell(N)}$ and any adaptive adversary $\mathcal{A}$ interacting with a $t(N)$-bounded leakage oracle $\mathcal{O}_{t(N)}(\cdot)$ for $t(N)=N\tau$,
    \begin{equation}\label{eq: adaptive privacy}
    \mathsf{SD}\left (\mathsf{View}_\mathcal{A}^{\mathcal{O}_{t(N)}(\mathsf{Share}(\mathsf{s}_0))};\mathsf{View}_\mathcal{A}^{\mathcal{O}_{t(N)}(\mathsf{Share}(\mathsf{s}_1))}\right )\leq \varepsilon(N).
    \end{equation}
   \end{itemize}
When $\varepsilon(N)=0$, we say the SSS has perfect privacy. 
\end{itemize}
The difference $\gamma=\rho-\tau$ is called the relative gap, since $N\gamma=r(N)-t(N)$ is the threshold gap of the scheme. When clear from context, we write $\varepsilon,\delta,t,k,\ell$ instead of $\varepsilon(N),\delta(N),t(N),r(N),\ell(N)$. When the parameters are not specified, we call a $(\varepsilon,\delta)$-SSS simply a binary SSS.
\end{definition}



In the above definition, a binary SSS has a pair of designed relative thresholds $(\tau,\rho)$. In this work, we are concerned with constructing nearly-threshold binary SSS, namely, binary SSS with arbitrarily small relative gap $\gamma=\rho-\tau$. We also want our binary SSS to share a large secret $\ell=\Omega(N)$.



\begin{definition} 
For any $0\leq\tau<\rho\leq1$,
a coding rate $R\in[0,1]$ is achievable if there exists a family of $(\varepsilon,\delta)$-SSS with relative threshold pair $(\tau,\rho)$ such that $\varepsilon$ and $\delta$ are both negligible in $N$ and $\frac{\ell}{N}\rightarrow R$.
The highest achievable coding rate of binary SSS for a pair $(\tau,\rho)$ is called its capacity.

\end{definition}

By relating binary SSS 
to Wyner wiretap codes with a pair of BEC's, 
we obtain the following coding rate upper bound for binary SSS.
\begin{lemma}\label{lem: upperbound}
For $0\leq\tau<\rho\leq1$, the coding rate capacity of binary SSS with relative threshold pair $(\tau,\rho)$ is asymptotically upper-bounded by $\rho-\tau$.
\end{lemma}

\begin{proof}

Let $(\mathsf{Share},\mathsf{Recst})$ be a non-adaptive binary SSS with relative threshold pair $(\tau,\rho)$. We use $\mathsf{Share}$ as the encoder and $\mathsf{Recst}$ as the decoder, and verify in the following that we obtain a Wyner wiretap code for a BEC$_{p_m}$ main channel and a BEC$_{p_w}$ wiretapper channel, where $p_m=1-\rho-\xi$ and $p_w=1-\tau+\xi$, respectively, for arbitrarily small $\xi>0$. 
Erasure in binary SSS is worst-case, while it is probabilistic in the Wyner wiretap model.   We however note that asymptotically,  the number of random erasures of BEC$_{p_m}$ and BEC$_{p_w}$ approaches $Np_m$ and $Np_w$, respectively, with overwhelming probability, and so a code that protects against worst-case erasure can be used as a wiretap code with probabilistic erasure.  
In our proof we also take into account the difference in the  secrecy notion in SSS and in  the case of Wyner wiretap code.

 The $N$-bit output $\mathsf{Y}=\mathsf{Y}_1,\ldots,\mathsf{Y}_N$ of a BEC$_p$ has a distribution where each bit is identically independently erased with probability $p$. 
By the Chernoff-Hoeffding bounds, the fraction $\eta$ of erasures satisfies the following. For arbitrarily small $\xi>0$,
$$
\left\{
\begin{array}{ll}
\mathsf{Pr}[\eta\geq p+\xi]&\leq \left(\left(\frac{p}{p+\xi}\right)^{p+\xi}     \left(\frac{1-p}{1-p-\xi}\right)^{1-p-\xi}\right)^{N};\\
\mathsf{Pr}[\eta\leq p-\xi]&\leq \left(\left(\frac{p}{p-\xi}\right)^{p-\xi}     \left(\frac{1-p}{1-p+\xi}\right)^{1-p+\xi}\right)^{N}.\\
\end{array}
\right.
$$
Applying the two inequalities to BEC$_{p_m}$ and BEC$_{p_w}$, respectively, we obtain the following conclusions. The probability that BEC$_{p_m}$ has at most $p_m+\xi=1-\rho$ fraction of erasures and the probability that BEC$_{p_w}$ has at least $p_w-\xi=1-\tau$ fraction of erasures are
both at most $\exp(-\Omega(N))$ for arbitrarily small $\xi>0$.

We are ready to prove the Wyner wiretap reliability and secrecy properties as defined in \cite{Wyner,CK78}.

We show correct decoding with probability $1-o(1)$. 
When the erasures are below $p_m+\xi=1-\rho$ fraction, it follows directly from the reconstructability of SSS that the decoding error is bounded from above by $\delta$, which is arbitrarily small for big enough $N$, where the probability is over the randomness of the encoder. 
When the erasures are not below $p_m+\xi=1-\rho$ fraction, we do not have correct decoding guarantee. But as argued above, this only occurs with a negligible probability over the randomness of the BEC$_{p_m}$. 
Averaging over the channel randomness of the BEC$_{p_m}$, we have correct decoding with probability $1-o(1)$.

We show random message equivocation secrecy $\mathsf{H}(\mathsf{S}|\mathsf{W})\geq \ell(1-o(1))$, where $\mathsf{S}$ is a uniform secret and $\mathsf{W}=\mathsf{BEC}_{p_w}(\mathsf{Share}(\mathsf{S}))$ is the view of the wiretapper. 
We in fact first prove the wiretap indistinguishability security as defined in \cite{btv12} and then deduce that it implies Wyner wiretap secrecy as defined in \cite{Wyner,CK78}. For each of the erasure patterns (say $A\subset[N]$ are not erased) of BEC$_{p_w}$ that exceeds $p_w-\xi=1-\tau$ fraction (equivalently, $|A|\leq N\tau$), the binary SSS privacy gives that for any two secrets, the corresponding views $\mathsf{W}|(\mathsf{S}=\mathsf{s}_0,A \mbox{ not erased})$ and $\mathsf{W}|(\mathsf{S}=\mathsf{s}_1,A\mbox{ not erased})$ are indistinguishable with error $\varepsilon$, which is arbitrarily small for big enough $N$. The distribution $(\mathsf{W}|\mathsf{S}=\mathsf{s}_0)$ and $(\mathsf{W}|\mathsf{S}=\mathsf{s}_1)$ are convex combinations of $\mathsf{W}|(\mathsf{S}=\mathsf{s}_0,A \mbox{ not erased})$ and $\mathsf{W}|(\mathsf{S}=\mathsf{s}_1,A\mbox{ not erased})$, respectively, for all the erasure patterns $A$ of BEC$_{p_w}$. As argued before, the probability that the erasures does not exceed $p_w-\xi=1-\tau$ fraction is negligible. We average over the channel randomness of the wiretapper channel BEC$_{p_w}$ and claim that the statistical distance of $(\mathsf{W}|\mathsf{S}=\mathsf{s}_0)$ and $(\mathsf{W}|\mathsf{S}=\mathsf{s}_1)$ is arbitrarily small for big enough $N$. According to \cite{btv12}, this is strictly stronger than the Wyner wiretap secrecy. 


Finally we use the coding rate upper bound of the Wyner wiretap code to bound the coding rate of binary SSS. We have shown that a binary SSS with relative threshold pair $(\tau,\rho)$ is a wiretap code for the pair (BEC$_{p_m}$,BEC$_{p_w}$). According to \cite{Wyner,CK78}, the achievable coding rate for the Wyner wiretap code is $(1-p_m)-(1-p_w)=p_w-p_m=\rho-\tau+2\xi$. 
Since this holds for arbitrarily small $\xi>0$, we obtain an upper bound of $\rho-\tau$ for binary SSS with relative threshold pair $(\tau,\rho)$. 
\end{proof}

In the rest of the paper, we give two constant rate constructions of nearly-threshold binary SSS against non-adaptive adversary and adaptive adversary, respectively. The non-adaptive adversary construction is optimal in the sense that the coding rate achieves the upper bound in Lemma \ref{lem: upperbound}.

\section{Secret sharing against non-adaptive adversaries}\label{sec: capacity}


We first present our construction of capacity-achieving binary SSS against non-adaptive adversaries, using linear strong seeded extractors and optimal rate stochastic erasure correcting codes. The following theorem describes the construction using these components.

\begin{theorem}\label{th: wiretap} 
Let $\mathsf{Ext} \colon \{0,1\}^d\times\{0,1\}^n\rightarrow\{0,1\}^\ell$ be a linear strong seeded $(n-t,\frac{\varepsilon}{8})$-extractor and $\mathsf{Ext^{-1}}(\mathsf{z},\cdot) \colon \{0,1\}^\ell\times\mathcal{R}_1\rightarrow\{0,1\}^n$ be the inverter of the function $\mathsf{Ext}(\mathsf{z},\cdot)$ that maps an $\mathsf{s}\in\{0,1\}^\ell$ to one of its pre-images chosen uniformly at random. 
Let $(\mathsf{SA\mbox{-}ECCenc},\mathsf{SA\mbox{-}ECCdec})$ be a stochastic affine-erasure correcting code with the encoder $\mathsf{SA\mbox{-}ECCenc}\colon\{0,1\}^{d+n}\times\mathcal{R}_2\rightarrow\{0,1\}^N$ that tolerates $N-r$ erasures and decodes with success probability at least $1-\delta$.   
Then the following coding scheme $(\mathsf{Share},\mathsf{Recst})$ is a non-adaptive $(\varepsilon,\delta)$-SSS with threshold pair $(t,r)$. 
\begin{align*}
\left\{
\begin{array}{ll}
\mathsf{Share}(\mathsf{s})&=\mathsf{SA\mbox{-}ECCenc(Z||Ext^{-1}}(\mathsf{Z},\mathsf{s})),\mathrm{where}\  \mathsf{Z}\stackrel{\$}{\leftarrow}\{0,1\}^d;\\
\mathsf{Recst}(\tilde{\mathsf{v}})&=\mathsf{Ext(\mathsf{z},\mathsf{x}), \mathrm{where}\  (\mathsf{z}||\mathsf{x})=SA\mbox{-}ECCdec}(\tilde{\mathsf{v}}).
\end{array}
\right.
\end{align*}
Here $\tilde{\mathsf{v}}$ denotes an incomplete version of a share vector $\mathsf{v}\in\{0,1\}^N$ with some of its components replaced by erasure symbols.
\end{theorem}

The proof of Theorem \ref{th: wiretap} will follow naturally from Lemma \ref{th: extractor property}.  We first state and prove this general property of a linear strong extractor, which is of independent interest. 
For the property to hold, we in fact only need the extractor to be able to extract from affine sources. %
The proof of Lemma \ref{th: extractor property} is a bit long. We then break it into a claim and two propositions.
\begin{lemma}\label{th: extractor property}
Let $\ext\colon \zo^d \times \zo^n \to \zo^m$ be a linear strong $(k,\varepsilon)$-extractor. Let $f_A\colon \zo^{d+n} \to \zo^t$ be any affine function with output length $t\leq n-k$. For any $\mathsf{m},\mathsf{m}'\in\{0,1\}^m$, let $(\mathsf{Z},\mathsf{X})=(\mathsf{U}_d,\mathsf{U}_n)|\left(\mathsf{Ext}(\mathsf{U}_d,\mathsf{U}_n)=\mathsf{m}\right)$ and $(\mathsf{Z}',\mathsf{X}')=(\mathsf{U}_d,\mathsf{U}_n)|\left(\mathsf{Ext}(\mathsf{U}_d,\mathsf{U}_n)=\mathsf{m}'\right)$. We have 
\begin{equation}\label{eq: pairwise}
\mathsf{SD}(f_A(\mathsf{Z},\mathsf{X});f_A(\mathsf{Z}',\mathsf{X}'))\leq 8\varepsilon.
\end{equation}
\end{lemma}

\begin{proof}
For the above pairwise guarantee (\ref{eq: pairwise}) to hold, it suffices to show that for every fixed choice of $\mathsf{m} \in \zo^m$,
the distribution of $f_A(\mathsf{Z},\mathsf{X})$ is $(4\varepsilon)$-close to $\mathsf{U}_d \times \cD$,
where $\mathsf{U}_d$ is the uniform distribution on $\zo^d$.

Without loss of generality, we assume that the linear function $\ext(\mathsf{z},\cdot):\zo^n \to \zo^m$, for every seed $\mathsf{z}$, has the entire $\zo^m$ as its image 
\footnote{If this condition is not satisfied for some choice $\mathsf{z}$ of the seed, there must be linear dependencies between the $m$ output bits of $\ext(\mathsf{z},\cdot)$. Therefore, for this choice $\ext(\mathsf{z},\cdot)$ can never be an extractor and arbitrarily changing $\ext(\cdot, \mathsf{z})$ to be an arbitrary full rank linear function will not change the overall performance of the extractor.}
.
Without loss of generality, it suffices to assume that $f_A$ is of the form
$f_A(\mathsf{Z},\mathsf{X})=(\mathsf{Z},W(\mathsf{X}))$ for some affine function $W\colon \zo^n \to \zo^t$.
This is because for any arbitrary $f_A$, the information contained in $f_A(\mathsf{Z},\mathsf{X})$
can be obtained from $(\mathsf{Z},W(\mathsf{X}))$ for a suitable choice of $W$.
Let $\cD$ be the uniform distribution on the image of $W$.

Let $\mathsf{K} \leftarrow \zo^n$ be a random variable uniformly distributed over
the kernel of the linear transformation defined by $W$,
and note that it has entropy at least $n-t\geq k$. 
The extractor $\ext$ thus guarantees that $\ext(\mathsf{Z},\mathsf{K})$, for a uniform
and independent seed $\mathsf{Z}$, is $\varepsilon$-close to uniform. 
By averaging, it follows that for at least $1-4\eps$ fraction
of the choices of the seed $\mathsf{z} \in \zo^d$, the distribution
of $\ext(\mathsf{z},\mathsf{K})$ is $(1/4)$-close to uniform. 

\begin{claim}\label{claim}
Let $\mathsf{U}$ be uniformly distributed on $\zo^m$ and
$\mathsf{U}'$ be any affine source that is not uniform on $\zo^m$.
Then, the statistical distance between $\mathsf{U}$ and $\mathsf{U}'$ is
at least $1/2$. 
\end{claim}

Claim \ref{claim} follows from the observation that any affine source $\mathsf{U}'$ that is not uniform on $\zo^m$ will have a support (the set of vectors $\mathsf{u}$ such that $\mathsf{Pr}[\mathsf{U}'=\mathsf{u}]>0$) that is an affine subspace of $\zo^m$ with dimension at most $m-1$.

Continuing with the previous argument, since $\ext$ is a linear function for every seed, 
the distribution of $\ext(\mathsf{z},\mathsf{K})$ for any seed $\mathsf{z}$ is
an affine source.
Therefore, the above claim allows us to conclude that for at least
$1-4\eps$ fraction of the choices of $\mathsf{z}$, the distribution
of $\ext(\mathsf{z},\mathsf{K})$ is \emph{exactly} uniform. Let $\cG \subseteq \zo^d$
be the set of such choices of the seed. Observe that if
$\ext(\mathsf{z},\mathsf{K})$ is uniform for some seed $\mathsf{z}$, then for any affine
translation of $\mathsf{K}$, namely, $\mathsf{K}+\mathsf{v}$ for any $\mathsf{v} \in \zo^n$,
we have that $\ext(\mathsf{z},\mathsf{K}+\mathsf{v})$ is uniform as well. This is
due to the linearity of the extractor.

Recall that our goal is to show that $f_A(\mathsf{Z},\mathsf{X})=(\mathsf{Z},W(\mathsf{X}))$ is $(4\varepsilon)$-close to $\mathsf{U}_d \times \cD$.
The distribution $(\mathsf{Z},W(\mathsf{X}))$ is obtained as $(\mathsf{U}_d,W(\mathsf{U}_n))|(\mathsf{Ext}(\mathsf{U}_d,\mathsf{U}_n)=\mathsf{m})$. For the rest of the proof, we first find out the distribution $(\mathsf{U}_d,W(\mathsf{U}_n))|(\mathsf{Ext}(\mathsf{U}_d,\mathsf{U}_n)=\mathsf{m},\mathsf{U}_d=\mathsf{z})$ for a seed $\mathsf{z}\in\cG$ (Proposition \ref{prop:WSX}) and then take the convex combination over the uniform seed to obtain $(\mathsf{Z},W(\mathsf{X}))$ (Proposition \ref{prop:SWX}). The argument starts with a uniform message $\mathsf{M}\stackrel{\$}{\leftarrow} \zo^m$ instead of a particular message $\mathsf{m}\in \zo^m$. We define a new set of simplified notations. Let $\mathsf{Z}\stackrel{\$}{\leftarrow} \zo^d$ be an independent and uniform seed. Let $(\mathsf{Z},\mathsf{Y})$ be the pre-image of $\mathsf{M}$ and $(\mathsf{Z},\mathsf{W})\colon=f_A(\mathsf{Z},\mathsf{Y})$. Proposition \ref{prop:WSX} and \ref{prop:SWX} are stated in terms of the triple $(\mathsf{M},\mathsf{Z},\mathsf{W})$.

\begin{prop} \label{prop:WSX}
Let $\mathsf{z} \in \cG$ and consider any $\mathsf{m} \in \zo^m$.
Then, the conditional distribution of $\mathsf{W}|(\mathsf{Z}=\mathsf{z}, \mathsf{M}=\mathsf{m})$ 
is exactly $\cD$. 
\end{prop}

To prove Proposition \ref{prop:WSX}, note that the distribution of $(\mathsf{Z},\mathsf{Y})$ is uniform on $\zo^{d+n}$. 
Now, fix any $\mathsf{z}\in \cG$ and
let $\mathsf{w} \in \zo^t$ be any element in the image of $W(\cdot)$.
Since the conditional distribution $\mathsf{Y}|(\mathsf{Z}=\mathsf{z})$ is uniform over $\zo^n$, further conditioning on $W(\mathsf{Y})=\mathsf{w}$ yields that
$\mathsf{Y}|(\mathsf{Z}=\mathsf{z},\mathsf{W}=\mathsf{w})$ is uniform 
over a translation of the kernel of $W(\cdot)$. By the assumption $\mathsf{z} \in \cG$ and recalling
$\mathsf{M}=\ext(\mathsf{Z},\mathsf{Y})$, we therefore
know that the extractor output is exactly uniform over $\zo^m$. That is, $\mathsf{M}|(\mathsf{Z}=\mathsf{z}, \mathsf{W}=\mathsf{w})$
is exactly uniform over $\zo^m$ and hence in this case $\mathsf{M}$ and $\mathsf{W}$ are independent. On the other hand, the distribution of $(\mathsf{Z},\mathsf{W})$ is exactly $\mathsf{U}_d \times \cD$, since the map $W(\cdot)$ is linear.
In particular, for any $\mathsf{z} \in \zo^d$, the conditional distribution $\mathsf{W}|(\mathsf{Z}=\mathsf{z})$ is exactly $\cD$.
This together with the fact that $\mathsf{M}$ and $\mathsf{W}$ are independent yield that the conditional
distribution of $(\mathsf{M},\mathsf{W})|(\mathsf{Z}=\mathsf{z})$ is exactly $\mathsf{U}_m \times \cD$.
We have therefore proved Proposition \ref{prop:WSX}.


\begin{prop}\label{prop:SWX}
For any $\mathsf{m} \in \zo^m$, the conditional distribution of $(\mathsf{Z},\mathsf{W})|(\mathsf{M}=\mathsf{m})$
is $(4\eps)$-close to $\mathsf{U}_d \times \cD$.
\end{prop}

To prove Proposition \ref{prop:SWX}, it suffices to note that the distribution of $(\mathsf{Z},\mathsf{W})|(\mathsf{M}=\mathsf{m})$ is a convex
combination of the distributions $(\mathsf{Z},\mathsf{W})|(\mathsf{M}=\mathsf{m},\mathsf{Z}=\mathsf{z})$ and then
use the result of Proposition~\ref{prop:WSX} along with the
fact that $\Pr[\mathsf{Z}\notin \cG]\leq 4\varepsilon$. A detailed derivation follows.

Recall that for any $\mathsf{z} \in \zo^d$, the conditional distribution
of $\mathsf{W}|(\mathsf{Z}=\mathsf{z})$ is exactly $\cD$ (since $\mathsf{Y}|(\mathsf{Z}=\mathsf{z})$ is uniform over
$\zo^n$). Consider any event $\cE \subseteq \zo^{d+t}$ and
let $p := \Pr[(\mathsf{Z},\mathsf{W}) \in \cE]$. Since $\mathsf{Z}$ and $\mathsf{W}$ are independent,
we have that
\[
p = 2^{-d} \sum_{(\mathsf{z},\mathsf{w}) \in \cE} \cD(\mathsf{w}), 
\]
where $\cD(\mathsf{w})$ denotes the probability assigned to the outcome $\mathsf{w}$ by $\cD$.
On the other hand, we shall write down the same probability in the conditional
probability space $\mathsf{M}=\mathsf{m}$ and show that it is different from $p$ by at most $4\eps$,
concluding the claim on the statistical distance. 
We have
\begin{align*}
\Pr[(\mathsf{Z},\mathsf{W}) \in \cE|\mathsf{M}=\mathsf{m}] &= \sum_{(\mathsf{z},\mathsf{w}) \in \cE} \Pr[\mathsf{Z}=\mathsf{z},\mathsf{W}=\mathsf{w}|\mathsf{M}=\mathsf{m}] \\
&= \sum_{(\mathsf{z},\mathsf{w}) \in \cE, \mathsf{z} \in \cG} \Pr[\mathsf{Z}=\mathsf{z},\mathsf{W}=\mathsf{w}|\mathsf{M}=\mathsf{m}] \\
&\ \ +
\sum_{(\mathsf{z},\mathsf{w}) \in \cE, \mathsf{z} \notin \cG} \Pr[\mathsf{Z}=\mathsf{z},\mathsf{W}=\mathsf{w}|\mathsf{M}=\mathsf{m}].
\end{align*}
Note that
\begin{align*}
\eta := \sum_{(\mathsf{z},\mathsf{w}) \in \cE, \mathsf{z} \notin \cG} \Pr[\mathsf{Z}=\mathsf{z},\mathsf{W}=\mathsf{w}|\mathsf{M}=\mathsf{m}] \leq
\Pr[\mathsf{Z} \notin \cG|\mathsf{M}=\mathsf{m}] \leq 4\eps,
\end{align*}
since $\mathsf{M}$ and $\mathsf{Z}$ are independent. 
Therefore, 
\begin{align}
\Pr[(\mathsf{Z},\mathsf{W}) \in \cE|\mathsf{M}=\mathsf{m}] &= \sum_{(\mathsf{z},\mathsf{w}) \in \cE, \mathsf{z} \in \cG} \Pr[\mathsf{Z}=\mathsf{z},\mathsf{W}=\mathsf{w}|\mathsf{M}=\mathsf{m}] + \eta \nonumber \\
&= 2^{-d} \sum_{(\mathsf{z},\mathsf{w}) \in \cE, \mathsf{z} \in \cG} \Pr[\mathsf{W}=\mathsf{w}|\mathsf{M}=\mathsf{m},\mathsf{Z}=\mathsf{z}] +\eta \label{eqn:prop:WSX:ind} \\
&= 2^{-d} \sum_{(\mathsf{z},\mathsf{w}) \in \cE, \mathsf{z} \in \cG} \cD(\mathsf{w}) +\eta \label{eqn:prop:WSX} \\
&= 2^{-d} \Big(\sum_{(\mathsf{z},\mathsf{w}) \in \cE} \cD(\mathsf{w})-\sum_{(\mathsf{z},\mathsf{w}) \in \cE, \mathsf{z} \notin \cG} \cD(\mathsf{w})\Big) +\eta \nonumber
\end{align}
where \eqref{eqn:prop:WSX:ind} uses the independence of $\mathsf{W}$ and $\mathsf{Z}$ and
\eqref{eqn:prop:WSX} follows from Proposition~\ref{prop:WSX}.
Observe that
\[
\eta' := 2^{-d}\sum_{(\mathsf{z},\mathsf{w}) \in \cE, \mathsf{z} \notin \cG} \cD(\mathsf{w}) = 2^{-d}\sum_{\mathsf{z} \notin \cG} 
\sum_{\substack{\mathsf{w}\colon \\ (\mathsf{z},\mathsf{w})\in \cE}} \cD(\mathsf{w})
\leq 2^{-d}(2^d-|\cG|) \leq 4 \eps.
\]
Therefore,
\begin{align*}
\Pr[(\mathsf{Z},\mathsf{W}) \in \cE|\mathsf{M}=\mathsf{m}] &= p+\eta-\eta' = p \pm 4\eps = \Pr[(\mathsf{Z},\mathsf{W}) \in \cE] \pm 4\eps,
\end{align*}
since $0 \leq \eta \leq 4\varepsilon$ and $0 \leq \eta' \leq 4\varepsilon$.
We have therefore proved Proposition \ref{prop:SWX}.
\end{proof}

With Lemma \ref{th: extractor property} at hand, we are now at a good position to prove Theorem \ref{th: wiretap}. 

\begin{proof}[Proof of Theorem \ref{th: wiretap}]
The reconstruction from $r$ shares follows trivially from the definition of stochastic erasure correcting code. We now prove the privacy.

The sharing algorithm of the SSS (before applying the stochastic affine code) takes a secret, which is a particular extractor output $\mathsf{s}\in\{0,1\}^\ell$, and uniformly samples a seed $\mathsf{z}\in\{0,1\}^d$ of $\mathsf{Ext}$ before uniformly finds an $\mathsf{x}\in\{0,1\}^n$ such that $\mathsf{Ext}(\mathsf{z},\mathsf{x})=\mathsf{s}$. This process of obtaining $(\mathsf{z},\mathsf{x})$ is the same as sampling $(\mathsf{U}_d,\mathsf{U}_n)\stackrel{\$}{\leftarrow}\{0,1\}^{d+n}$ and then restrict to $\mathsf{Ext}(\mathsf{U}_d,\mathsf{U}_n)=\mathsf{s}$. We define the random variable tuple 
\begin{equation}\label{eq: pre-image}
(\mathsf{Z},\mathsf{X}):=(\mathsf{U}_d,\mathsf{U}_n)|\left(\mathsf{Ext}(\mathsf{U}_d,\mathsf{U}_n)=\mathsf{s}\right)
\end{equation} 
and refer to it as the pre-image of $\mathsf{s}$.

Let $\Pi_A:\{0,1\}^N\rightarrow\{0,1\}^t$ be the projection function that maps a share vector to the $t$ shares with index set $A\subset[N]$ chosen by the non-adaptive adversary. Observe that the combination $(\Pi_A\circ\mathsf{SA\mbox{-}ECCenc}):\{0,1\}^{d+n}\rightarrow\{0,1\}^t$ (for any fixed randomness $\mathsf{r}$ of $\mathsf{SA\mbox{-}ECCenc}$) is an affine function. 
So the view of the adversary is simply the output of the affine function $f_A=(\Pi_A\circ\mathsf{SA\mbox{-}ECCenc})$ applied to the random variable tuple $(\mathsf{Z},\mathsf{X})$ defined in (\ref{eq: pre-image}). 


We can now formulate the privacy of the SSS in this context. We want to prove that the statistical distance of the views of the adversary for a pair of secrets $\mathsf{s}$ and $\mathsf{s}'$ can be made arbitrarily small. The views of the adversary are the outputs of the affine function $f_A$ with inputs $(\mathsf{Z},\mathsf{X})$ and $(\mathsf{Z}',\mathsf{X}')$ for the secret $\mathsf{s}$ and $\mathsf{s}'$, respectively. 
According to Lemma \ref{th: extractor property}, we then have that the privacy error is $8\times\frac{\varepsilon}{8}=\varepsilon$.
\end{proof}

\remove{
\begin{remark}\label{rmk: seed attack}
Note that the same argument cannot be made if the set $A$ is chosen according to the seed $\mathsf{z}$, in which case the source (induced by the set $A$) of $\mathsf{Ext}$ depends on its seed $\mathsf{z}$, violating the definition of a seeded extractor. To see this, recall that the source of $\mathsf{Ext}$ is uniform $n$-bit conditioned on the output of the affine function $f_A=(\Pi_A\circ\mathsf{SA\mbox{-}ECCenc})$. The choice of $A$ then can affect this source.
In a real life adaptive attack, the adversary can first spend some reading budget on figuring out the value of the seed $\mathsf{z}$ of $\mathsf{Ext}$, and then decide the rest of the reading positions according to the seed value. Finding out the value of $\mathsf{z}$ is possible because the encoding randomness of the $\mathsf{SA\mbox{-}ECCenc}$ is included as part of the share vector. The adversary can choose to read the encoding randomness $\mathsf{r}$ of the $\mathsf{SA\mbox{-}ECCenc}$. After $\mathsf{r}$ is revealed, the $\mathsf{SA\mbox{-}ECCenc}$ is just a deterministic affine function. It is then possible to recover information about $\mathsf{z}$ by solving linear equations.
\end{remark}
}


We now analyze the coding rate of the $(\varepsilon,\delta)$-SSS with relative threshold pair $(\frac{t}{N},\frac{r}{N})$ constructed in Theorem \ref{th: wiretap} when instantiated with the $\mathsf{SA\mbox{-}ECC}$ from Lemma \ref{th: SA-ECC} and the $\mathsf{Ext}$ from Lemma \ref{lem: Ext}. The secret length is $\ell=n-t-O(d)$, where the seed length is $d=O(\log^3(2n/\varepsilon))$. The $\mathsf{SA\mbox{-}ECC}$ encodes $d+n$ bits to $N$ bits and with coding rate $R_{ECC}=\rho-\xi$ for a small $\xi$ determined by $\delta$ (satisfying the relation $\delta=\exp(-\Omega(\xi^2N/\log^2 N))$ according to Lemma \ref{th: SA-ECC}). We then have $n=N(\rho-\xi)-d$, resulting in the coding rate   
\begin{align*}
R=\frac{\ell}{N}=\frac{n-t-O(d)}{N}=\frac{N(\rho-\xi)-t-O(d)}{N}=\rho-\tau-(\xi+\frac{O(d)}{N})
=\rho-\tau-o(1).
\end{align*}

\begin{corollary}\label{cor: non-adaptive} For any $0\leq\tau<\rho\leq1$, there is an explicit construction of non-adaptive $(\varepsilon,\delta)$-SSS with relative threshold pair $(\tau,\rho)$ achieving coding rate $\rho-\tau-o(1)$, where $\varepsilon$ and $\delta$ are both negligible. 
\end{corollary}

The binary SSS obtained in Corollary \ref{cor: non-adaptive} is asymptotically optimal as it achieves the upper bound in Lemma \ref{lem: upperbound}.

\section{Secret sharing against adaptive adversaries}\label{sec: construction}


In this section, we will describe our construction which achieves privacy against adaptive adversaries, using seedless affine extractors. We start with the specific extraction property needed from our affine extractors.

\begin{definition}\label{def: APExt} An affine extractor $\mathsf{AExt}\colon\{0,1\}^n\rightarrow\{0,1\}^m$ 
is called $(k,\varepsilon)$-almost perfect if for any affine $(n,k)$-source $\mathsf{X}$,
\begin{align*}
\left |\mathsf{Pr}[\mathsf{AExt}(\mathsf{X})=\mathsf{y}]-\frac{1}{2^m}\right |\leq 2^{-m}\cdot\varepsilon, \mbox{ for any }\mathsf{y}\in \{0,1\}^m.
\end{align*} 
\end{definition}
Almost perfect property can be trivially achieved by requiring an exponentially (in $m$) small error in statistical distance, using the relation between $\ell_\infty$-norm and $\ell_1$-norm.

\begin{theorem}\label{th: generic}
Let $\mathsf{AExt}\colon\{0,1\}^n\rightarrow\{0,1\}^\ell$ be an invertible 
$(n-t,\frac{\varepsilon}{2})$-almost perfect affine extractor and $\mathsf{AExt^{-1}}\colon\{0,1\}^\ell\times\mathcal{R}_1\rightarrow\{0,1\}^n$ be its $v$-inverter that maps an $\mathsf{s}\in\{0,1\}^\ell$ to one of its pre-images chosen uniformly at random. 
Let $(\mathsf{SA\mbox{-}ECCenc},\mathsf{SA\mbox{-}ECCdec})$ be a stochastic affine-erasure correcting code with encoder $\mathsf{SA\mbox{-}ECCenc}\colon\{0,1\}^n\times\mathcal{R}_2\rightarrow\{0,1\}^N$ that tolerates $N-r$ erasures and decodes with success probability at least $1-\delta$. 
Then the $(\mathsf{Share},\mathsf{Recst})$ defined as follows is an adaptive $(\varepsilon+v,\delta)$-SSS with threshold pair $(t,r)$.
\begin{align*}
\left\{
\begin{array}{ll}
\mathsf{Share}(\mathsf{s})&=\mathsf{SA\mbox{-}ECCenc(AExt^{-1}}(\mathsf{s}));\\
\mathsf{Recst}(\tilde{\mathsf{v}})&=\mathsf{AExt(SA\mbox{-}ECCdec}(\tilde{\mathsf{v}})),
\end{array}
\right.
\end{align*}
where $\tilde{\mathsf{v}}$ denotes an incomplete version of a share vector $\mathsf{v}\in\{0,1\}^N$ with some of its components replaced by erasure symbols.
\end{theorem}

\begin{proof} The $(r,\delta)$-reconstructability 
follows directly from the erasure correcting capability of the $\mathsf{SA\mbox{-}ECC}$. For any 
$\tilde{\mathsf{v}}$ with at most $N-r$ erasure symbols and the rest of its components consistent with a valid codeword $\mathsf{v}\in\{0,1\}^N$, the $\mathsf{SA\mbox{-}ECC}$ decoder identifies the unique codeword $\mathsf{v}$ with probability $1-\delta$ over the encoder randomness. The corresponding $\mathsf{SA\mbox{-}ECC}$ message of $\mathsf{v}$ is then inputted to $\mathsf{AExt}$ and the original secret $\mathsf{s}$ is reconstructed with the same probability. 

We next prove the $(t,\varepsilon)$-privacy. Without loss of generality, we first assume the inverter of the affine extractor is perfect, namely, $v=0$. When $v$ is negligible but not equal to zero, the overall privacy error will increase slightly, but still remain negligible.
For any $\mathsf{r}\in\mathcal{R}_2$, the affine encoder of $\mathsf{SA\mbox{-}ECC}$ is characterised by a matrix $G_\mathsf{r}\in\{0,1\}^{n\times N}$ and an offset $\Delta_\mathsf{r}$. For $n$ unknowns $\mathbf{x}=(x_1,\ldots,x_n)$, we have
\begin{align*}
\mathsf{SA\mbox{-}ECCenc(\mathbf{x})}=\mathbf{x}G_\mathsf{r}+\Delta_\mathsf{r}=(\mathbf{x}G_1,\ldots,\mathbf{x}G_N)+\Delta_\mathsf{r},
\end{align*}
where $G_i=(g_{1,i},\ldots,g_{n,i})^T$ (here the subscript ``$_\mathsf{r}$'' is omitted to avoid double subscripts) denotes the $i$th column of $G_\mathsf{r}$, $i=1,\ldots, N$. This means that knowing a component $c_i$ of the $\mathsf{SA\mbox{-}ECC}$ codeword is equivalent to obtaining a linear equation $c_i\oplus\Delta_i=\mathbf{x}G_i=g_{1,i}x_1+\cdots+g_{n,i}x_n$  about the $n$ unknowns $x_1,\ldots,x_n$, where $\Delta_i$ (again, the subscript ``$_\mathsf{r}$'' is omitted) denotes the $i$th component of $\Delta_\mathsf{r}$. 
Now, we investigate the distribution of $\mathsf{View}_\mathcal{A}^{\mathcal{O}_t(\mathsf{Share}(\mathsf{s}))}$ for any secret $\mathsf{s}\in\{0,1\}^\ell$ by finding the probability $\mathsf{Pr}[\mathsf{View}_\mathcal{A}^{\mathcal{O}_t(\mathsf{Share}(\mathsf{s}))}=\mathsf{w}]$ for arbitrary $\mathsf{w}$. We then have 
\begin{align*}
\mathsf{Pr}[\mathsf{View}_\mathcal{A}^{\mathcal{O}_t(\mathsf{Share}(\mathsf{s}))}=\mathsf{w}]&=\mathsf{Pr}[\mathsf{View}_\mathcal{A}^{\mathcal{O}_t(\mathsf{SA\mbox{-}ECCenc}(\mathsf{X}))}=\mathsf{w}|\mathsf{AExt}(\mathsf{X})=\mathsf{s}]\\
                                                                                                                        &=\frac{\mathsf{Pr}[\mathsf{AExt}(\mathsf{X})=\mathsf{s}|\mathsf{View}_\mathcal{A}^{\mathsf{SA\mbox{-}ECCenc}(\mathsf{X})}=\mathsf{w}]\cdot\mathsf{Pr}[\mathsf{View}_\mathcal{A}^{\mathsf{SA\mbox{-}ECCenc}(\mathsf{X})}=\mathsf{w}]}{\mathsf{Pr}[\mathsf{AExt}(\mathsf{X})=\mathsf{s}]}\\
                                                                                                                        &\stackrel{(i)}{=}\frac{(1\pm\frac{\varepsilon}{2})2^{-\ell}\cdot\mathsf{Pr}[\mathsf{View}_\mathcal{A}^{\mathsf{SA\mbox{-}ECCenc}(\mathsf{X})}=\mathsf{w}]}{\mathsf{Pr}[\mathsf{AExt}(\mathsf{X})=\mathsf{s}]}\\
                                                                                                                        &\stackrel{(ii)}{=}\frac{(1\pm\frac{\varepsilon}{2})2^{-\ell}\cdot \frac{2^{n-\mathrm{rank}(\mathcal{A})}}{2^n}}{2^{-\ell}}\\
                                                                                               &=(1\pm\frac{\varepsilon}{2})\cdot 2^{-\mathrm{rank}(\mathcal{A})},            
\end{align*}
where notations $\mathsf{X}$, $\pm$, $\mathrm{rank}(\mathcal{A})$ and $(i),(ii)$ are explained as follows.
In above, we first use the fact that $\mathsf{Pr}[\mathsf{View}_\mathcal{A}^{\mathcal{O}_t(\mathsf{Share}(\mathsf{s}))}=\mathsf{w}]$ can be seen as the probability of uniformly selecting $\mathsf{X}$ from $\{0,1\}^n$, with the condition that $\mathsf{AExt}(\mathsf{X})=\mathsf{s}$.  This is true because the sets  $\mathsf{AExt}^{-1}(\mathsf{s})$ for all $\mathsf{s}$, partition $\{0,1\}^n$ and the rest follows from Definition \ref{def: invertible}.
The shorthand ``$y=1\pm\frac{\varepsilon}{2}$'' denotes ``$1-\frac{\varepsilon}{2}\leq y\leq 1+\frac{\varepsilon}{2}$''. The shorthand ``$\mathrm{rank}(\mathcal{A})$'' denotes the rank of the up to $t$ columns of $G$ corresponding to the index set $A$ adaptively chosen by $\mathcal{A}$. The equality $(i)$ follows from the fact that $\mathsf{AExt}$ is an $(n-t,2^{-(\ell+1)}\varepsilon)$-affine extractor and the uniform $\mathsf{X}$ conditioned on at most $t$ linear equations is an affine source with at least $n-t$ bits entropy. The equality $(ii)$ holds if and only if $\mathsf{w}$ is in the set $\{\mathsf{SA\mbox{-}ECCenc}(\mathsf{x})_A:\mathsf{x}\in\{0,1\}^n\}$. Indeed, consider $\mathsf{X}$ as unknowns for equations, the number of solutions to the linear system $\mathsf{SA\mbox{-}ECCenc}(\mathsf{X})_A=\mathsf{w}$ is either $0$ or equal to $2^{n-\mathrm{rank}(\mathcal{A})}$.

The distribution of $\mathsf{View}_\mathcal{A}^{\mathcal{O}_t(\mathsf{Share}(\mathsf{s}))}$ for any secret $\mathsf{s}$ is determined by the quantity $\mathrm{rank}(\mathcal{A})$, which is independent of the secret $\mathsf{s}$. 
Let $\mathsf{W}$ be the uniform distribution over the set $\{\mathsf{SA\mbox{-}ECCenc}(\mathsf{x})_A:\mathsf{x}\in\{0,1\}^n\}$. Then by the triangular inequality, we have 
\begin{align*}
\begin{array}{l}
\mathsf{SD}\left(\mathsf{View}_\mathcal{A}^{\mathcal{O}_t(\mathsf{Share}(\mathsf{s}_0))};\mathsf{View}_\mathcal{A}^{\mathcal{O}_t(\mathsf{Share}(\mathsf{s}_1))}\right)\\
\leq\mathsf{SD}\left(\mathsf{View}_\mathcal{A}^{\mathcal{O}_t(\mathsf{Share}(\mathsf{s}_0))};\mathsf{W}\right)+\mathsf{SD}\left(\mathsf{W};\mathsf{View}_\mathcal{A}^{\mathcal{O}_t(\mathsf{Share}(\mathsf{s}_1))}\right)\\
\leq \frac{\varepsilon}{2}+\frac{\varepsilon}{2}
=\varepsilon.
\end{array}
\end{align*}
When the inverter of the affine extractor is not perfect, the privacy error is upper bounded by $\varepsilon+v$. 
This concludes the privacy proof.
\end{proof}




There are explicit constructions of binary affine extractors that, given a constant fraction of entropy, outputs a constant fraction of random bits with exponentially small error (see Lemma~\ref{lem: affine}). 
There are known methods for constructing an invertible affine extractor $\mathsf{AExt}'$ from any affine extractor $\mathsf{AExt}$ such that the constant fraction output size and exponentially small error properties are preserved. A simple method is to let $\mathsf{AExt}'(\mathsf{U}_n||M):=\mathsf{AExt}(\mathsf{U}_n)\oplus M$ (see Appendix~\ref{apdx:OTP} for a discussion). 
This is summarized in the lemma below.

\begin{lemma} \label{lem:affineExt}
For any $\delta \in (0,1]$, there is an explicit seedless
$(\delta n,\eps)$-almost perfect affine extractor $\mathsf{AExt}\colon \zo^n \to \zo^m$
where $m = \Omega(n)$ and $\eps = \exp(-\Omega(n))$.
Moreover, there is an efficiently computable $\eps$-inverter for the extractor.
\end{lemma}

\begin{proof}
Let $\mathsf{AExt}\colon \zo^n \to \zo^m$ be Bourgain's affine extractor (Lemma~\ref{lem: affine}) 
for entropy rate
$\mu$, output length $m=\Omega(n)$, and achieving exponentially small error
$\eps=\exp(-\Omega(n))$. Using the one-time pad trick (Appendix~\ref{apdx:OTP}),
we construct an invertible variant achieving output length $m'=\Omega(m)=\Omega(n)$
and exponentially small error. Finally, we simply truncate the output length of the
resulting extractor to $m''=\Omega(m')=\Omega(n)$ bits so that the
closeness to uniformity measured by $\ell_\infty$ norm required for almost-perfect
extraction is satisfied. The truncated extractor is still invertible
since the inverter can simply pad the given input with random bits and
invoke the original inverter function. 
\end{proof}

It now suffices to instantiate Theorem~\ref{th: generic} with
the explicit construction of SA-ECC and the invertible affine extractor $\mathsf{AExt}$
of Lemma~\ref{lem:affineExt}.
Let $R_{ECC}$ denote the rate of the SA-ECC. 
Then we have $R_{ECC}=\frac{n}{N}$, where $n$ is the input length of the affine extractor $\mathsf{AExt}$ and $N$ is the number of players. The intuition of the construction in Theorem~\ref{th: generic} is that if a uniform secret is shared and conditioning on the revealed shares the secret still has a uniform distribution (being the output of a randomness extractor), then no information is leaked. In fact, the proof of Theorem~\ref{th: generic} above is this intuition made exact, with special care on handling the imperfectness of the affine extractor. So as long as the ``source'' of the affine extractor $\mathsf{AExt}$ has enough entropy, privacy is guaranteed. Here the ``source'' is the distribution $\mathsf{U}_n$ conditioned on the adversary's view, which is the output of a $t$-bit affine function. The ``source'' then is affine and has at least $n - \tau N=n(1-\frac{\tau}{R_{ECC}})$ bits of entropy. Now as long as $\tau<R_{ECC}$, using the $\mathsf{AExt}$ from Lemma \ref{lem: affine} (more precisely, an invertible affine extractor $\mathsf{AExt}':\{0,1\}^{n'}\rightarrow\{0,1\}^\ell$ constructed from $\mathsf{AExt}$) with $\mu=1-\frac{\tau}{R_{ECC}}$, a constant fraction of random bits can be extracted with exponentially small error. This says that privacy is guaranteed for $\tau\in[0,R_{ECC})$. 
The stochastic affine ECC in Lemma~\ref{th: SA-ECC} asymptotically achieves the rate $1-(1-\rho)=\rho$. We then have the following corollary. 

\begin{corollary}\label{cor: adaptive} For any $0\leq\tau<\rho\leq1$, there is an explicit constant coding rate adaptive $(\varepsilon,\delta)$-SSS with relative threshold pair $(\tau,\rho)$, where $\varepsilon$ and $\delta$ are both negligible.
\end{corollary}

The construction above achieves a constant coding rate for any $(\tau,\rho)$ pair satisfying $0\leq\tau<\rho\leq 1$. However, since the binary affine extractor in Lemma \ref{lem: affine} does not extract all the entropy from the source 
and moreover the step that transforms an affine extractor into an invertible affine extractor incurs non-negligible overhead, the coding rate of the above construction does not approach $\rho-\tau$. We leave explicit constructions of binary SSS against adaptive adversary with better coding rate as an interesting technical open question.

\section{Conclusion}\label{sec: conclusion}
We studied the problem of sharing  {\em arbitrary long secrets} using constant length shares and required a nearly-threshold access structure. By nearly-threshold we mean a ramp scheme with arbitrarily small gap to number of players ratio. We show that by replacing perfect privacy and reconstructibility with slightly relaxed notions and inline with similar strong cryptographic notions, one can explicitly construct such nearly-threshold schemes. We gave two constructions with security against non-adaptive and adaptive adversaries, respectively, and proved optimality of the former. Our work also make a new connection between secret sharing and wiretap coding. 

We presented our model and constructions for the extremal case
of binary shares. However, we point out that the model and
our constructions can be extended to shares over any desired
alphabet size $q$. Using straightforward observations (such
as assigning multiple shares to each player), this task reduces
to extending the constructions over any prime $q$. In this case,
the building blocks that we use; namely, the stochastic error-correcting
code, seeded and seedless affine extractors need to be extended
to the $q$-ary alphabet. The constructions 
\cite{Trevisan,improvement,GS10} that we use, however, can be extended to
general alphabets with straightforward modifications. The only
exception is Bourgain's seedless affine extractor \cite{Bourgain}. The extension of \cite{Bourgain} to arbitrary alphabets is not straightforward and has been
accomplished in a work by Yehudayoff \cite{ref:Yehudayoff2011}.

Our constructions are not linear: even the explicit  non-adaptive construction that   uses an affine
function for every fixing of the encoder's randomness  does not result in a linear secret sharing.
Linearity is an essential property in applications such as multiparty computation and so 
explicit constructions of {\em linear}  secret sharing schemes
in our model will be an important achievement. Yet another important direction for future work is  deriving bounds and constructing optimal codes for finite length ($N$) case.
Such result will also be of high interest for  wiretap coding.

\bibliographystyle{plain}
\bibliography{1.bib}

\begin{thebibliography}{10}

\bibitem{ALCP09}
Vaneet Aggarwal, Lifeng Lai, A.~Robert Calderbank, and H.~Vincent Poor.
\newblock Wiretap channel type {II} with an active eavesdropper.
\newblock {\em {IEEE} International Symposium on Information Theory, {ISIT}
  2009}, pages 1944--1948, 2009.

\bibitem{btv12}
Mihir Bellare, Stefano Tessaro, and Alexander Vardy.
\newblock Semantic security for the wiretap channel.
\newblock In {\em Advances in Cryptology - {CRYPTO} 2012}, volume 7417 of {\em
  Lecture Notes in Computer Science}, pages 294--311. Springer, 2012.

\bibitem{Blakley}
George~R. Blakley.
\newblock Safeguarding cryptographic keys.
\newblock In {\em Proceedings of the 1979 AFIPS National Computer Conference},
  pages 313--317, 1979.

\bibitem{BGK16}
Andrej Bogdanov, Siyao Guo, and Ilan Komargodski.
\newblock Threshold secret sharing requires a linear size alphabet.
\newblock In {\em Theory of Cryptography - {TCC} 2016-B}, pages 471--484, 2016.

\bibitem{BIVW16}
Andrej Bogdanov, Yuval Ishai, Emanuele Viola, and Christopher Williamson.
\newblock Bounded indistinguishability and the complexity of recovering
  secrets.
\newblock In {\em Advances in Cryptology - {CRYPTO}}, pages 593--618, 2016.

\bibitem{BW17}
Andrej Bogdanov and Christopher Williamson.
\newblock Approximate bounded indistinguishability.
\newblock In {\em International Colloquium on Automata, Languages, and
  Programming, {ICALP}}, pages 53:1--53:11, 2017.

\bibitem{Bourgain}
Jean Bourgain.
\newblock On the construction of affine extractors.
\newblock {\em Geometric and Functional Analysis}, 17(1):33--57, 2007.

\bibitem{ref:CCX13}
Ignacio {Cascudo Pueyo}, Ronald Cramer, and Chaoping Xing.
\newblock Bounds on the threshold gap in secret sharing and its applications.
\newblock {\em {IEEE} Trans. Information Theory}, 59(9):5600--5612, 2013.

\bibitem{AC0}
Kuan Cheng, Yuval Ishai, and Xin Li.
\newblock Near-optimal secret sharing and error correcting codes in {AC}0.
\newblock In {\em Theory of Cryptography - {TCC} 2017, Part {II}}, pages
  424--458, 2017.

\bibitem{nearoptimalwt}
Mahdi Cheraghchi.
\newblock Nearly optimal robust secret sharing.
\newblock {\em Designs, Codes and Cryptography}, pages 1--20, 2018.

\bibitem{CDS12}
Mahdi Cheraghchi, Fr{\'{e}}d{\'{e}}ric Didier, and Amin Shokrollahi.
\newblock Invertible extractors and wiretap protocols.
\newblock {\em {IEEE} Trans. Information Theory}, 58(2):1254--1274, 2012.

\bibitem{ref:CDB15}
Ronald Cramer, Ivan Damg{\aa}rd, and Jesper~Buus Nielsen.
\newblock {\em Secure Multiparty Computation and Secret Sharing}.
\newblock Cambridge University Press, 2015.

\bibitem{CDFP08}
Ronald Cramer, Yevgeniy Dodis, Serge Fehr, Carles Padr{\'{o}}, and Daniel
  Wichs.
\newblock Detection of algebraic manipulation with applications to robust
  secret sharing and fuzzy extractors.
\newblock In {\em Advances in Cryptology - {EUROCRYPT} 2008}, volume 4965 of
  {\em Lecture Notes in Computer Science}, pages 471--488. Springer, 2008.

\bibitem{CK78}
Imre Csisz{\'{a}}r and Janos K{\"o}rner.
\newblock Broadcast channels with confidential messages.
\newblock {\em {IEEE} Trans. Information Theory}, 24(3):339--348, 1978.

\bibitem{list-decode}
Venkatesan Guruswami.
\newblock List decoding from erasures: bounds and code constructions.
\newblock {\em {IEEE} Trans. Information Theory}, 49(11), 2003.

\bibitem{GS10}
Venkatesan Guruswami and Adam~D. Smith.
\newblock Optimal rate code constructions for computationally simple channels.
\newblock {\em J. {ACM}}, 63(4):35:1--35:37, 2016.

\bibitem{XinLiAffine}
Xin Li.
\newblock A new approach to affine extractors and dispersers.
\newblock {\em {IEEE} Conference on Computational Complexity, {CCC} 2011},
  pages 137--147, 2011.

\bibitem{OKT92}
Wakaha Ogata, Kaoru Kurosawa, and Shigeo Tsujii.
\newblock Nonperfect secret sharing schemes.
\newblock {\em Advances in Cryptology - {AUSCRYPT} '92}, pages 56--66, 1992.

\bibitem{OW84}
Lawrence~H. Ozarow and Aaron~D. Wyner.
\newblock Wire-tap channel {II}.
\newblock {\em The Bell System Technical Journal}, 63:2135--2157, Dec 1984.

\bibitem{improvement}
Ran Raz, Omer Reingold, and Salil~P. Vadhan.
\newblock Extracting all the randomness and reducing the error in {T}revisan's
  extractors.
\newblock {\em J. Comput. Syst. Sci.}, 65(1):97--128, 2002.

\bibitem{Shamir}
Adi Shamir.
\newblock How to share a secret.
\newblock {\em Commun. {ACM}}, 22(11):612--613, 1979.

\bibitem{SMITH07}
Adam~D. Smith.
\newblock Scrambling adversarial errors using few random bits, optimal
  information reconciliation, and better private codes.
\newblock In {\em ACM-SIAM Symposium on Discrete Algorithms SODA '07}, pages
  395--404. Society for Industrial and Applied Mathematics, 2007.

\bibitem{ref:Sti92}
Douglas~R. Stinson.
\newblock An explication of secret sharing schemes.
\newblock {\em Des. Codes Cryptography}, 2(4):357--390, 1992.

\bibitem{Trevisan}
Luca Trevisan.
\newblock Extractors and pseudorandom generators.
\newblock {\em J. {ACM}}, 48(4):860--879, 2001.

\bibitem{Wyner}
Aaron~D. Wyner.
\newblock The wire-tap channel.
\newblock {\em The Bell System Technical Journal}, 54(8):1355--1387, Oct 1975.

\bibitem{ref:Yehudayoff2011}
Amir Yehudayoff.
\newblock Affine extractors over prime fields.
\newblock {\em Combinatorica}, 31(2):245, 2011.

\end{thebibliography}

\appendix
\addcontentsline{toc}{section}{Appendices}
\renewcommand{\thesubsection}{\Alph{subsection}}


\section{Proof Sketch of Lemma \ref{th: SA-ECC}}\label{apdx: SA-ECC}
\begin{proof}
We defer the details of the construction to the full version of this paper. For now we refer to \cite[Theorem 6.1]{GS10} and point out the adaptations needed. There are six building blocks involved in the construction: $\mathsf{SC}$, $\mathsf{RS}$, $\mathsf{Samp}$, $\mathsf{KNR}$, $\mathsf{POLY}_t$ and $\mathsf{REC}$. We replace the first and last building blocks. 

The first building block is a {\em Stochastic Code} ($\mathsf{SC}$). We need two properties from this building block: detect (output $\bot$) when the codeword is masked by a random offset and correct from erasures of no more than $1-\rho+\varepsilon$ fraction. While the former property is always satisfied by the original $\mathsf{SC}$ used in \cite{GS10}, the latter property might not hold. When $1-\rho$ is small, we can let the decoder of the $\mathsf{SC}$ used in \cite{GS10} set the erased bits to $0$ and decode it as  errors. But when $1-\rho>\frac{1}{2}$, this trick no longer works. We 
then combine a systematic AMD in \cite{CDFP08} and an erasure list-decodable code \cite{list-decode} (sub-optimal suffices since $\mathsf{SC}$ encodes the control information which is negligible) to obtain the
$\mathsf{SC}$ in our construction of $\mathsf{SA\mbox{-}ECC}$.

The last building block is a {\em Random Error Code} ($\mathsf{REC}$). We also need two properties from this building block: correct from random erasures of $1-\rho$ fraction and the encoder is a linear function. We need the latter property for affine property of the $\mathsf{SA\mbox{-}ECC}$ constructed. Explicit linear codes at rate $1-p$ that correct $p$ fraction of random erasures are known. We can use any explicit construction of capacity achieving codes for BEC$_{1-\rho}$ for $\mathsf{REC}$ and use a similar argument of \cite{SMITH07}.

We now refer to the \textbf{Algorithm 1.} in the proof of \cite[Theorem 6.1]{GS10} and show that, with the $\mathsf{SC}$ and $\mathsf{REC}$ replaced accordingly, we do have a $\mathsf{SA\mbox{-}ECC}$. The error correction capability and optimal rate follow similarly as in the proof of \cite[Theorem 6.1]{GS10}. We next show affine property. \textbf{Phase 1} and \textbf{Phase 2} are about the {\em control information}, which are part of the encoding randomness $\mathsf{r}$ of the $\mathsf{SA\mbox{-}ECC}$ to be fixed to constant value in the analysis of affine property. During \textbf{Phase 3}, the message $\mathbf{m}$ is linearly encoded (our $\mathsf{REC}$ is linear) and then permuted, followed by adding a translation term $\Delta_\mathsf{r}$. Since permutation is a linear transformation, we combine the two linear transformations and write 
$\mathbf{m}G_\mathsf{r}+\Delta_\mathsf{r}$, where $G_\mathsf{r}$ is a binary matrix. 
Finally, during \textbf{Phase 4}, some blocks that contain the control information are inserted into $\mathbf{m}G_\mathsf{r}+\Delta_\mathsf{r}$. We add dummy zero columns into $G_\mathsf{r}$ and zero blocks into $\Delta_\mathsf{r}$ to the corresponding positions where the control information blocks are inserted. Let $\mathbf{m}\hat{G}_\mathsf{r}+\hat{\Delta}_\mathsf{r}$ be the vector after padding dummy zeros. Let $\hat{\Delta'}_\mathsf{r}$ be the vector obtained from padding dummy zero blocks, complementary to the padding above, to the control information blocks. We then write the final codeword of the $\mathsf{SA\mbox{-}ECC}$ in the form $\mathbf{m}\hat{G}_\mathsf{r}+(\hat{\Delta}_\mathsf{r}+\hat{\Delta'}_\mathsf{r})$, which is indeed an affine function of the message $\mathbf{m}$.

\end{proof}

\section{One-Time-Pad trick of inverting extractors}\label{apdx:OTP}


There is a well known way to transform an efficient function into one that is also efficiently invertible
through a ``One-Time-Pad'' trick. We give a proof for the special case of affine extractors, for completeness.

\begin{lemma}\label{lem: general inverter} 
Let $\mathsf{AExt}\colon \{0,1\}^n\rightarrow\{0,1\}^m$ be an affine $(n,k)$-extractor with error $\varepsilon$. Then $\mathsf{AExt}': \{0,1\}^{n+m}\rightarrow\{0,1\}^m$ defined as follows is a $\varepsilon$-invertible affine $(n+m,k+m)$-extractor with error $\varepsilon$.
$$
\mathsf{AExt}'(\mathsf{z})=\mathsf{AExt}'(\mathsf{x}||\mathsf{y})=\mathsf{AExt}(\mathsf{x})+\mathsf{y},
$$
where the input $\mathsf{z}\in\{0,1\}^{n+m}$ is separated into two parts: $\mathsf{x}\in\{0,1\}^{n}$ and $\mathsf{y}\in\{0,1\}^{m}$. 
\end{lemma}

\begin{proof}
Let $\mathsf{Z}$ be a random variable with flat distribution supported on an affine subspace of $\{0,1\}^{n+m}$ of dimension at least $k+m$. Separate $\mathsf{Z}$ into two parts $\mathsf{Z}=(\mathsf{X}||\mathsf{Y})$, where $\mathsf{X}\in\{0,1\}^{n}$ and $\mathsf{Y}\in\{0,1\}^m$. Then conditioned on any $\mathsf{Y}=\mathsf{y}$, $\mathsf{X}$ has a distribution supported on an affine subspace of $\{0,1\}^{n}$ of dimension at least $k$. This asserts that conditioned on any $\mathsf{Y}=\mathsf{y}$,
$$
\mathsf{SD}(\mathsf{AExt}(\mathsf{X})+\mathbf{y};\mathsf{U}_{\{0,1\}^m})\leq \varepsilon.
$$
Averaging over the distribution of $\mathsf{Y}$ concludes the extractor proof. 

We next show an efficient inverter $\mathsf{AExt}'^{-1}$ for $\mathsf{AExt}'$. For any $\mathsf{s}\in\{0,1\}^m$, define
$$
\mathsf{AExt}'^{-1}(\mathsf{s})=(\mathsf{U}_n||\mathsf{AExt}(\mathsf{U}_n)+\mathsf{s}).
$$
The randomised function $\mathsf{AExt}'^{-1}$ is efficient and $\mathsf{AExt}'^{-1}(\mathsf{U}_m)\stackrel{\varepsilon}{\sim}\mathsf{U}_{n+m}$.
\end{proof}

\end{document}